\documentclass[twocolumn]{autart}    
\usepackage{amsfonts}
\usepackage{textcomp}
\usepackage{enumerate}
\usepackage{cite}
\usepackage{color,xcolor}
\usepackage{amsmath}
\usepackage{amssymb}
\usepackage{graphicx}
\usepackage{algorithm}
\usepackage{algpseudocode}
\usepackage{bbm}
\usepackage{graphics}
\usepackage{subfigure}
\usepackage{epsfig}
\usepackage{multicol} 
\usepackage{stfloats}
\usepackage{appendix}

\usepackage{wrapfig}

\newtheorem{theorem}{Theorem}
\newtheorem{proposition}{Proposition}
\newtheorem{lemma}{Lemma}

\newtheorem{remark}{Remark}

\newtheorem{assumption}{Assumption}
\newenvironment{proof}{\hspace{0ex}\textsc{Proof}.\hspace{1ex}}{\hfill$\Box$\newline}
\makeatletter

\newcommand{\Rmnum}[1]{\expandafter\@slowromancap\romannumeral #1@}
\makeatother                                   

\begin{document}

\begin{frontmatter}

\title{On the Effects  of Modeling Errors on Distributed Continuous-time Filtering
} 


\author[address_a]{Xiaoxu Lyu}\ead{eelyuxiaoxu@ust.hk},
\author[address_b]{Shilei Li}\ead{shileili@bit.edu.cn},
\author[address_b]{Dawei Shi}\ead{daweishi@bit.edu.cn},  
\author[address_a]{Ling Shi}\ead{eesling@ust.hk}

\address[address_a]{Department of Electronic and Computer Engineering, Hong Kong University of Science and Technology, Clear Water
Bay, Kowloon, Hong Kong, China}  

\address[address_b]{School of Automation, Beijing Institute
of Technology, Beijing 100081, China}

%


\begin{abstract}                
This paper  offers a comprehensive performance  analysis
 of the distributed continuous-time filtering  in the presence of  modeling errors.
First,  we  introduce  two performance indices, namely the nominal performance index and the estimation error covariance.
By  leveraging   the nominal performance index and the Frobenius  norm
of the modeling deviations,  we derive  the bounds of the estimation error covariance and the lower bound of the nominal performance index. Specifically, we reveal the effect of  the consensus parameter on both bounds.
 We demonstrate that, under specific  conditions, an incorrect process noise
covariance can lead to the divergence of the estimation
error covariance. Moreover,  we  investigate the properties  of the eigenvalues of the error dynamical matrix.
Furthermore,   we explore  the magnitude  relations between the nominal performance index and the estimation error covariance. Finally, we present some numerical simulations to validate  the effectiveness of the theoretical results.

\end{abstract}

\end{frontmatter}

Sensor networks  have been extensively  employed in  various
 domains, including   satellite  navigation \cite{morton2021position},  environmental monitoring \cite{muduli2018application}, and   robot cooperative mapping \cite{tian2022kimera}.
 Distributed state estimation   \cite{olfati2007distributed,olfati2009kalman,battistelli2014kullback,battistelli2014consensus} plays a  significant role in sensor networks.
 It  is  utilized  to  estimate  the system state for  each intelligent   sensor within sensor networks  by  leveraging   the collective   information from its own observations   as well as  the shared  knowledge  from  its  local neighbors.

  %

Distributed state estimation can be classified into two primary  categories: distributed  continuous-time filter  and distributed  discrete-time filter,   which are  based on  two  distinct types of  state space model frameworks.
In comparison to   the distributed discrete-time filters, fewer  research  advancements     have   been achieved   in the stochastic continuous-time model setting, primarily due to its inherent complexity.
Some results of distributed continuous-time  state estimation  are presented as follows.
Olfati-Saber \cite{olfati2007distributed}  proposed  a distributed  continuous-time  Kalman filter by  integrating    the  consensus fusion  term  of the state estimate  into the Kalman-Bucy filter \cite{kalman1961new}.  It  is important to note that  a crucial requirement for this  distributed filter  was the availability of   the prior knowledge of the initial state
 and covariance, and  its stability was proven  under   the assumption of  a  noise-free scenario.  To tackle the issue  of unguaranteed  boundedness of covariance matrices  in the  absence of each sensor's observability\cite{olfati2007distributed},   Kim et al. \cite{kim2016distributed}  proposed  algorithms by exchanging  covariance matrices  among sensors.  However, the stability and the parameter conditions were also derived   under   the noise-free assumption.
 Similarly to the filter structure  presented in  \cite{olfati2007distributed}, Wu et al. \cite{wu2016consensus}  utilized   the inverse of  covariance matrices as   weights  for  the consensus terms.  Nonetheless, this method encountered  similar issues  as discussed above.
Battilotti et~al. \cite{battilotti2020asymptotically}  exhibited an asymptotically optimal  distributed filter, and  provided   a comprehensive analysis of   its   convergence    and  optimality.
In this paper, the fundamental  filter framework  for performance  analysis in the presence of  modeling errors will be  based on the  distributed continuous-time filter  presented  in   \cite{battilotti2020asymptotically}. This choice is motivated by the less restrictive  stability  conditions  and  the asymptotically  optimal properties  exhibited by  this filter.
  It is worth mentioning  that the analysis methods employed in this framework  can be adapted to other filters with ease.

In practical scenarios,  obtaining precise models is rarely feasible, and   modeling errors are  widespread.
The modeling errors  involve  deviations  in the state matrix, deviations in the measurement matrix, and      mismatched noise covariances.   Such  modeling errors have the potential  to result in a decline in the filter performance and  even lead to  system failure.   Several studies  have  explored the performance of  continuous-time filters  for a single sensor when confronted with  modeling errors.
T.~Nishimura \cite{nishimura1967error} established a conservative design criterion  to  ensure    the upper bound of the
estimation error covariance  for  continuous-time systems with  incorrect noise covariances.
Griffin and Sage \cite{griffin1968large}
conducted a thorough  sensitivity analysis of
filtering and smoothing algorithms,  exploring  both large and small-scale scenarios,  in the presence of   the modeling  errors.
 Fitzgerald   \cite{fitzgerald1971divergence}    investigated    the  divergence of the Kalman filter,  highlighting the   mean square error may become unbounded due to the incorrect process noise covariance.
  Toda and Patel    \cite{toda1978performance}  derived    performance  index bounds  and   mean   square error  bounds    for  the suboptimal filters, considering  the presence of modeling errors.
Sangsuk-Iam and Bullock  \cite{sangsuk1988analysis} analyzed
the behavior of  the   continuous-time Kalman filter  under mismatched noise covariances.
Additionally,  some existing literature  on  state estimation  addresses   arbitrary noise scenarios \cite{teng2024gmkf,granichin2004linear,li2000discrete}.
 However,  these studies primarily focus on the single-sensor systems.  When  shifting the focus to the  distributed filters,  numerous challenges arise.  These challenges include  dealing with  the  new distributed  filter structure,  addressing the coupling  terms,  and handling the behavior  of  the consensus parameter.
 This paper will investigate  the effect of the modeling errors on the distributed continuous-time  filters.

Motivated by the  aforementioned   observations,  the primary objective of  this paper  is to offer a comprehensive  performance analysis  of the distributed  continuous-time filter in the presence of modeling errors. This analysis aims to enrich individuals'   comprehension   and  prognostication  of the filters' behavior.
The main contributions of this paper are summarized   as  follows:

\begin{enumerate}
\item
We define  two performance indices, namely  the nominal performance index and the estimation error covariance, as metrics  to   assess  the performance  of the filter. By utilizing the nominal performance index and the norm of the modeling deviations,
   the bounds of   the estimation error covariance are determined (\textbf{Theorem~\ref{theorem bounds estimation error covariance}}).
In particular, the asymptotic decay behavior of the bounds
in terms of {\color{blue} the nominal consensus parameter} $\gamma_u$ is   revealed.
    These results   play an important role in evaluating
   the estimation error covariance by utilizing  the nominal models in the presence of  the modeling errors.

\item
 We reveal that  an  incorrect process noise covariance   can lead to the divergence of the  estimation error covariance  for undirected connected  networks,  regardless of the magnitude of the consensus parameter (\textbf{Theorem~\ref{theorem incorrect Q divergence}}).
It is demonstrated  that   the nominal   parameters    must  be reasonably designed to  ensure   the convergence of the estimation error covariance.

\item
The relations between the nominal performance index and the estimation error covariance, including their magnitude  relations and  the norm bound of their difference,    are analyzed based on the characteristics  of  the noise covariance deviations (\textbf{Theorem~\ref{theorem relation sigmau sigmaa}}). This result offers insights into the judicious   selection of the nominal covariance to ensure a  conservative estimation.
\end{enumerate}

The subsequent sections of this paper are structured as follows.
Section~\ref{sec pre and problem fomulation} provides  the  problem statement.  Section~\ref{sec performance indices} presents
  the nominal distributed filter  and the  estimation error covariance.  Section~\ref{sec performance analysis} focuses on analyzing    bounds of the estimation error covariances,  divergence,  and
   relations between different performance  indices.  Section~\ref{sec simulations}
 exhibits some examples  to validate the effectiveness of the  theoretical results.
  Section~\ref{sec conclusion}  concludes this  paper.

\textit{Notations}: Throughout this paper, define
$\mathcal{R}^n$  and  $\mathcal{R}^{n\times m}$ as the sets of  $n$-dimensional real vectors and
 $n\times m$-dimensional real matrices, respectively.   The notation $\text{Re}(\cdot)$  represents  the operation of taking the real part.  The vector $1_N$ is an $N$-dimensional vector, where all its elements are equal to $1$. For a matrix $A\in \mathcal{R}^{n\times m}$,  let $\Vert A\Vert_F$  and $\Vert A\Vert_2$ represent the Frobenius norm and the spectral norm, respectively, and  $A^T$ and $A^{-1}$  denote its transpose and inverse, respectively.
Let $\sigma_1(A)\geq \cdots \geq \sigma_n(A)$ denote the decreasingly ordered  singular values of the matrix  $A$, and  $\bar \sigma(A)$  and $\underline \sigma(A)$ represent the maximum and  minimum singular values, respectively. Similarly,  let  $\lambda_1(A)\geq \cdots \geq \lambda_n(A)$ denote  the decreasingly ordered eigenvalues of a Hermitian matrix $A$, and $\bar \lambda(A)$ and $\underline \lambda(A)$  denote the maximum and minimum eigenvalues of the matrix $A$.
Let $\alpha(A) = \text{Re}(\bar \lambda(A))$ be  the spectral abscissa of  the matrix $A$.
Moreover, the matrix inequality $A>B$ ($A\geq B$) means that $A-B$ is positive define (positive semi-definite). The symbol $\otimes$  denotes the Kronecker product,  $[A]_{ij}$  represents   the $(i,j)$-th element of the matrix $A$, $\mathbb{E}\{x\}$ denotes  the expectation of the random variable $x$,  $\text{Tr}(A)$ is the trace, and  $\text{vec}(A)$ refers to the column vector  obtained  by concatenating the columns of the matrix $A$.
 The logarithmic norm is defined as  $\mu(A) = \lim_{h\to 0^{+}}\frac{\| I+hA\|-1}{h}$, where   $h$ is a real, positive number,  $\|\cdot\|$ is an induced matrix norm, and  $I$ is the identity matrix  of the same dimension as $A$.
{\color{blue}
The three most common logarithmic  norms    are  $\mu_1(A) = \mathop{\text{sup}}\limits_{j}(\text{Re}([A]_{jj})+\sum_{i\neq j}\vert [A]_{ij}\vert)$,   $\mu_2(A) = \bar\lambda(\frac{A+A^T}{2})$, and
$\mu_{\infty}(A) = \mathop{\text{sup}}\limits_{i}(\text{Re}([A]_{ii})+\sum_{j\neq i}\vert [A]_{ij}\vert)$.  The $abs(\cdot)$ function  represents the operation of taking  the absolute value of a number.
 }

\section{Problem Statement}\label{sec pre and problem fomulation}

Consider a continuous-time linear stochastic system, measured by a sensor network of  $N$ sensors, described by
\begin{equation}\label{eq dynamics}
\begin{aligned}
&\dot x(t) = Ax(t)+\omega(t),\\
& y_{i}(t) = C_{i}x(t)+\nu_{i}(t), ~~~i=1,2,...,N,
\end{aligned}
\end{equation}
where  $t$  is the time index, $x(t)\in \mathcal{R}^n$ is the system state, $y_i(t)\in \mathcal{R}^{m_i}$ is the  measurement  of sensor $i$, $A\in \mathcal{R}^{n\times n}$ is the state  matrix,
$C_i \in \mathcal{R}^{m_i \times n}$ is the measurement matrix of sensor $i$,
$\omega(t)\in \mathcal{R}^{n}$ is the process noise,  and $\nu_i(t)\in \mathcal{R}^{m_i}$ is the measurement noise of sensor $i$.
It is assumed that  $\omega(t)$ and $\nu_i(t)$ are zero-mean
white noise processes, and these noise processes    are uncorrelated with each other and with the initial state $x(0)$, i.e.,
$\mathbb{E}\{ \omega(t)\omega^T(\tau) \}  = Q\delta(t-\tau)$, $\mathbb{E}\{ \nu_i(t)\nu^T_i(\tau) \}  = R_i\delta(t-\tau), ~i= 1,\ldots, N,$  $\mathbb{E}\{ \omega(t)\nu^T_i(\tau) \}  = 0, ~ i= 1,\ldots, N,$  and  $\mathbb{E}\{ \nu_i(t)\nu^T_j(\tau) \}  = 0, ~i, j =1,\ldots, N,i\neq j$,  $Q\in \mathcal{R}^{n\times n}$ and $R_i\in \mathcal{R}^{m_i \times m_i}$ denote the process noise covariance and the measurement noise covariance of sensor $i$, respectively,   and  $\delta(t)$ represents  the Dirac delta function,  given by  $\delta(t-\tau) = 1$ if  $t=\tau$ and 0 otherwise.
The communication topology of sensor networks is described by
using the graph theory, as detailed in Appendix~\ref{sec graph theory}.

%
%

We adopt a classical distributed continuous-time  Kalman filter in~\cite{battilotti2020asymptotically}, given  by
\begin{equation}\label{eq standard filter}
\begin{aligned}
\dot{\hat x}_i(t)  =& A \hat x_i(t) + K_i(y_{i}(t)-C_{i}\hat x_i(t)) \\
&+ \gamma P(\infty) \sum_{j\in \mathcal{N}_i} (\hat x_{j}(t) - \hat x_{i}(t)),
\end{aligned}
\end{equation}
where
$K_i = NP(\infty)C^T_{i}R^{-1}_{i}$,  {\color{blue}  $\gamma$ is the consensus parameter},  $P(\infty)$ is the solution of
\begin{equation}\label{eq steady Ps Lyapunove euqaiton}
\begin{aligned}
0 =& A P(\infty) + P(\infty)A^T + Q- P(\infty)C^T_{c}R^{-1}_{d}C_{c}P(\infty),
\end{aligned}
\end{equation}
$C_{c} = [C^T_{1},\ldots, C^T_{N}]^T$  is the augmented  measurement matrix, and   $R_{d} = \text{diag}\{R_{1},\ldots, R_{N}\}$ is the augmented measurement noise covariance.


%

Accurate model parameters, such as  $A$, $C_i$, $Q$, and $R_i$, are   often  difficult  to obtain  in practice.  Therefore,   we utilize  the nominal model parameters, denoted as $A_u$, $C_{i,u}$,  $Q_u$, and  $R_{i,u}$,  in the practical filter.  The deviations  between the nominal parameters and the actual parameters are defined  as
\begin{equation}
\begin{aligned}
&\Delta A = A_{u}-A, ~~~\Delta C_i = C_{i,u} - C_{i},\\
&\Delta Q = Q_{u} - Q,~~~\Delta R_i = R_{i,u} - R_i.~~
\end{aligned}
\end{equation}
For these deviations,  we may have some prior knowledge or estimates, such as  the Frobenius norm bound of these deviations and    the magnitude relations between the nominal  and  actual parameters. The objective of this paper  is to analyze the effect of these modeling errors on the performance of the distributed  filter and provide  performance  estimates based on the available  prior information.
The problem of the effects of the modeling errors on distributed continuous-time  filtering  is  formulated as follows:
\begin{enumerate}
\item  Based on the  prior Frobenius norm bound for the deviations between the nominal and actual parameters,  we evaluate the performance of the distributed filter  by utilizing the known nominal  model parameters.
\item  Explain how   an incorrect  process noise covariance can lead to the divergence of  the distributed filter,  even when the system model is accurate.
    Investigate   the relations between different performance indices in terms of the consensus gain.
\end{enumerate}

\section{Performance Indices}\label{sec performance indices}
This  section   provides two performance indices, namely  the nominal  performance index and the estimation error covariance, and
the  estimation error covariance  is derived
based on  the nominal distributed  Kalman   filter.

\subsection{Nominal Distributed Kalman Filter}
In practice,  we utilize the nominal parameters   $A_u$, $C_{i,u}$,  $Q_u$, and   $R_{i,u}$ in the  distributed  filter.
Similarly to the ideal  distributed Kalman filter~(\ref{eq standard filter}), the nominal distributed Kalman filter with known  and inaccurate nominal  parameters can be expressed as
\begin{equation}\label{eq dot x_iu}
\begin{aligned}
~~~~~~\dot{\hat x}_{i,u}(t)  =& A_u \hat x_{i,u}(t) + K_{i,u}(y_{i}(t)-C_{i,u}\hat x_{i,u}(t)) \\
&+ \gamma_u P_u(\infty) \sum_{j\in \mathcal{N}_i} (\hat x_{j,u}(t) - \hat x_{i,u}(t)),
\end{aligned}
\end{equation}
where $K_{i,u} = NP_u(\infty)C^T_{i,u}R^{-1}_{i,u}$,   {\color{blue}  $\gamma_u$ is the nominal  consensus parameter}, $P_u(\infty)$ is the solution of
\begin{equation}\label{eq steady Pu Lyapunove euqaiton}
\begin{aligned}
0 =& A_u P_u(\infty) + P_u(\infty)A^T_u + Q_u\\
 &- P_u(\infty)C^T_{c,u}R^{-1}_{d,u}C_{c,u}P_u(\infty),
\end{aligned}
\end{equation}
$C_{c,u} = [C^T_{1,u},\ldots, C^T_{N,u}]^T$ is the augmented nominal measurement matrix,
and $R_{d,u} = \text{diag}\{R_{1,u},\ldots, R_{N,u}\}$ is the augmented nominal measurement noise covariance.
%
For further analysis,  some notations  are  defined.
The subscript $\text{`} u \text{'}$ represents the nominal counterpart, the subscript $\text{`}c\text{'}$ denotes the column vector, and the subscript $\text{`}d\text{'}$ indicates the  diagonal matrix. Define  $A_d = \text{diag}\{A, \ldots, A\}$,
$A_{d,u}= \text{diag}\{A_u,\ldots,A_u\}$, $C_{d} = \text{diag}\{C_{1},\ldots, C_{N}\}$, $C_{c,u} = [C^T_{1,u},\ldots, C^T_{N,u}]^T$, $C_{d,u} = \text{diag}\{C_{1,u},\ldots, C_{N,u}\}$, $R_{c,u} = [R^T_{1,u},\ldots, R^T_{N,u}]^T$, $R_{d,u} = \text{diag}\{R_{1,u},\ldots, R_{N,u}\}$, $R_{c} = [R^T_{1},\ldots, R^T_{N}]^T$, $R_{d} = \text{diag}\{R_{1},\ldots, R_{N}\}$,
$y_c(t) = [y^T_1(t),\ldots, y^T_N(t)]^T$, $\nu_c(t) = [\nu^T_1(t),\ldots,\nu^T_N(t)]^T$, $\hat x_{c,u}(t) = [\hat x^T_{1,u}(t),\ldots, \hat x^T_{N,u}(t)]^T$,  and
$\dot{\hat x}_{c,u}(t) = [\dot{\hat x}^T_{1,u}(t),\ldots, \dot{\hat x}^T_{N,u}(t)]^T$.

The compact form of  (\ref{eq dot x_iu})   can be written as
\begin{equation}\label{eq mathcal Au}
\begin{aligned}
\dot{\hat x}_{c,u}(t) = \mathcal{A}_u \hat x_{c,u}(t) + K_{d,u}y_c(t),
\end{aligned}
\end{equation}
where  $G_{i,u}=A_u-K_{i,u}C_{i,u}$,  $G_{d,u} = \text{diag}\{G_{1,u},\ldots,G_{N,u}\}$,  $K_{d,u} = \text{diag}\{K_{1,u},\ldots,K_{N,u}\}$, and
\begin{equation}\label{eq mathcal Au to s}
\begin{aligned}
\mathcal{A}_u =  G_{d,u} - \gamma_u (\mathcal{L}\otimes P_u(\infty)).~~
\end{aligned}
\end{equation}
For the nominal distributed filter,  we define $\Sigma_u(t)$  as  the nominal performance index, which evaluates the filter's   performance based on the nominal parameters:
\begin{equation}\label{eq sigmau dot time}
\begin{aligned}
~~\dot \Sigma_{u}(t)=& \mathcal{A}_u \Sigma_{u}(t) + \Sigma_{u}(t)\mathcal{A}^T_u \\
&+ K_{d,u}R_{d,u} K^T_{d,u} +U_{N}\otimes Q_{u},
\end{aligned}
\end{equation}
where  $U_N = 1_N1^T_N$.
The nominal performance index  is derived  under the  assumption  that
there are no deviations in the nominal parameters. It  has the same form as the estimation error covariance presented  in \cite{battilotti2020asymptotically}.

\begin{remark}
In practical engineering applications, the consensus gain can be directly adjusted in  the code.
As the  consensus gain increases, indicating that more information  is being  utilized from other filters, the ideal distributed filter is expected to  achieve  better performance.
However, due to modeling errors,  the effect of  the consensus gain requires  further  exploration  in the following.
\end{remark}

%

%

\subsection{Estimation Error Covariance}

This subsection derives the bounds of   the estimation error covariance by utilizing the nominal performance index and the norm of the modeling deviations, and presents
the lower bound  of  the nominal performance index. Particularly,  the effect of the consensus parameter  on both bounds   is   revealed.

 The subscript $\text{`} e \text{'}$ is utilized to  represent the actual performance index.
The estimation error    and  the estimation error covariance are defined as  $\eta_{i}(t) = x(t)-\hat x_{i,u}(t)$  and
$\Sigma_{i,e}(t)  = \mathbb{E}\{ \eta_{i}(t)\eta^T_{i}(t) \}$,
respectively. Then, the dynamics of $\eta_i(t)$ and $\Sigma_{i,e}(t)$  can be computed as  $\dot \eta_{i}(t) = \dot x(t)- \dot {\hat x}_{i,u}(t)$ and  $\dot \Sigma_{i,e}(t)  = \mathbb{E}\{\dot \eta_{i}(t)\eta^T_{i}(t) + \eta_{i}(t)\dot \eta^T_{i}(t) \},$
respectively.  Define the augmented  estimation error   as   $\eta_c(t) = [\eta^T_{1}(t),\ldots, \eta^T_{N}(t)]^T$,     the   augmented estimation error covariance  as $\Sigma_e(t) = \mathbb{E}\{\eta_c(t)\eta^T_c(t)\}$,
$S(t) = \mathbb{E}\{ \eta_c(t)(1_N\otimes x(t))^T\}$,   and    $X(t)= \mathbb{E}\{ (1_N\otimes x(t))(1_N\otimes x(t))^T\}$.
Then,   their expressions are presented   as follows.



\begin{proposition}\label{proposition estimaiton error and covarnacne}
Consider the system~(\ref{eq dynamics}) and the nominal distributed  filter~(\ref{eq dot x_iu}).
The estimation error, the augmented estimation error, and the augmented estimation error covariance  can be expressed as follows:
\begin{enumerate}
\item
The estimation error is
\begin{equation}\label{eq estimation error}
\begin{aligned}
~~~~\dot \eta_{i}(t) & =  F_{i,u}x(t)  + G_{i,u}\eta_i(t) +\omega(t) -K_{i,u}\nu_i(t) \\ &~~~~-\gamma_uP_u(\infty)\sum_{j\in \mathcal{N}_i} (\eta_{i}(t)-\eta_{j}(t)),
\end{aligned}
\end{equation}
where  $F_{i,u} = A-A_u - K_{i,u}(C_i-C_{i,u})$.
\item
The augmented estimation error is
\begin{equation}\label{eq augmented estimation error}
\begin{aligned}
\dot \eta_c(t) & =  F_{d,u}(1_{N}\otimes x(t)) + \mathcal{A}_u \eta_c(t)\\
&~~~~+1_{N}\otimes \omega(t) - K_{d,u}\nu_c(t),~~~~~
\end{aligned}
\end{equation}
where $F_{d,u} = \text{diag}\{ F_{1,u},\ldots,F_{N,u}\}$.
\item
The augmented  estimation error covariance is
\begin{equation}\label{eq augmented estimation error covariance}
\begin{aligned}
\dot \Sigma_e(t)
& = \mathcal{A}_u \Sigma_e(t) +  \Sigma_e(t)\mathcal{A}^T_u + F_{d,u}S^T(t) \\
&~~~~+ S(t)F^T_{d,u}+ K_{d,u}R_d K^T_{d,u} +U_{N}\otimes Q,
\end{aligned}
\end{equation}
where    $U_N = 1_N1^T_N$,
$\dot S(t) = \mathcal{A}_uS(t) +  S(t)A^T_{d} + F_{d,u}X(t) +U_{N}\otimes Q,$  and  $\dot X(t) = A_{d}X(t)+X(t)A^T_{d}+U_{N}\otimes Q$.


\end{enumerate}

\end{proposition}

\begin{proof}
Item 1) and Item 2):  For  computational simplicity, the time index is omitted. By combining  (\ref{eq dynamics}) with  (\ref{eq dot x_iu}), one  has
\begin{equation}\label{eq proposi derive dot eta}
\begin{aligned}
\dot \eta_{i}& = \dot x-\dot {\hat x}_{i,u} \\
&=  (Ax+\omega)-  \Big(A_u \hat x_{i,u} + K_{i,u}(y_{i}-C_{i,u}\hat x_{i,u}) \\
&~~~~+ \gamma_u P_u(\infty) \sum_{j\in \mathcal{N}_i} (\hat x_{j,u} - \hat x_{i,u})\Big)    \\
& = (A-A_u)x+A_u(x-\hat x_{i,u})+\omega\\
&~~~~-K_{i,u}C_{i,u}(x-\hat x_{i,u})- K_{i,u}(C_i-C_{i,u})x\\
&~~~~-K_{i,u}\nu_i -\gamma_uP_u(\infty)\sum_{j\in \mathcal{N}_i} ((x-\hat x_{i,u})-(x-\hat x_{j,u}))\\
& = (A-A_u)x+A_u\eta_i+\omega-K_{i,u}C_{i,u}\eta_i-K_{i,u}\nu_i\\
&~~~~- K_{i,u}(C_i-C_{i,u})x -\gamma_uP_u(\infty)\sum_{j\in \mathcal{N}_i} (\eta_{i}-\eta_{j}).
\end{aligned}
\end{equation}
%
%
According to the definition of $F_{i,u}$ and  $G_{i,u}$,
(\ref{eq proposi derive dot eta}) can be rewritten as   (\ref{eq estimation error}).
Next, using  the definition of $F_{d,u}$, $G_{d,u}$,  and $K_{d,u}$ above (\ref{eq mathcal Au to s}),
 the augmented estimation error  can be expressed  as:
\begin{equation*}
\begin{aligned}
\dot \eta_c(t) & =  F_{d,u}(1_{N}\otimes x(t)) + G_{d,u} \eta_c(t)\\
&~~~~+1_{N}\otimes \omega(t) - K_{d,u}\nu_c(t)-\gamma_u (\mathcal{L}\otimes P_u(\infty))\eta_c(t).
\end{aligned}
\end{equation*}
Finally, by denoting $\mathcal{A}_u =  G_{d,u} - \gamma_u (\mathcal{L}\otimes P_u(\infty))$,  (\ref{eq augmented estimation error}) can be obtained.

Item 3): Define $\xi = [\eta^T_c, (1_{N}\otimes x)^T]^T$ \cite{griffin1968large}, and it follows $\dot \xi = \mathcal{F}\xi + \mathcal{B}\rho,$
where   $\rho =  [\nu^T_c,(1\otimes\omega)^T]^T$,
\begin{equation}
\begin{aligned}
\mathcal{F} = \left[ \begin{array}{cc}
\mathcal{A}_u & F_{d,u}\\
0 & A_{d}
\end{array}\right],~
\end{aligned}
\end{equation}
and
\begin{equation}
\begin{aligned}
\mathcal{B} = \left[ \begin{array}{cc}
 -K_{d,u} & I\\
0 & I
\end{array}\right].
\end{aligned}
\end{equation}

Define  $\Sigma_{\xi} = \mathbb{E}\{\xi\xi^T \}$ and $\Phi = \mathbb{E}\{\rho\rho^T\}$, and one has
\begin{equation}\label{eq sigmaxi}
\begin{aligned}
\dot \Sigma_{\xi} 
= \mathcal{F} \Sigma_{\xi} +  \Sigma_{\xi}\mathcal{F}^T + \mathcal{B}  \Phi \mathcal{B}^T,
\end{aligned}
\end{equation}
and
\begin{equation}
\begin{aligned}
\Phi = \left[ \begin{array}{cc}
 R_d & 0\\
0 & U_N\otimes Q
\end{array}\right].~~~~~
\end{aligned}
\end{equation}

Finally, the term $\dot\Sigma_e$ can be obtained based on (\ref{eq sigmaxi}), and the detailed expression is given in  (\ref{eq augmented estimation error covariance}).

\end{proof}

%

%
%

It is worth mentioning that if $F_{d,u}=0$,  the estimation error covariance (\ref{eq augmented estimation error covariance}) is simplified as
\begin{equation}\label{eq sigmaa fdu=0}
\begin{aligned}
\dot \Sigma_e(t)
& = \mathcal{A}_u \Sigma_e(t) +  \Sigma_e(t)\mathcal{A}^T_u + K_{d,u}R_d K^T_{d,u} +U_{N}\otimes Q.
\end{aligned}
\end{equation}

\section{Performance Analysis}\label{sec performance analysis}

This section  provides   an in-depth  performance analysis for the distributed  continuous-time   filtering from multiple different  aspects:  bounds of the estimation error covariance, divergence analysis,  and relations between the different  performance indices.
First, some assumptions are presented here for further analysis.

 \begin{assumption}\label{ass communication graph}
 The communication graph is undirected and  connected.
 \end{assumption}

\begin{assumption}\label{ass detectable uncontrolled}
$(A_u,C_{c,u})$  is  observable,  and  $(A_u, Q^{1/2}_u)$  is controllable.
\end{assumption}


%

\begin{proposition}\cite{battilotti2020asymptotically}\label{proposi gammau0}
Consider the system~(\ref{eq dynamics}) and the nominal distributed  filter~(\ref{eq dot x_iu}). Suppose that  Assumptions~\ref{ass communication graph} and \ref{ass detectable uncontrolled} hold.     There exists  $\bar \gamma_{u_0}$ such that for all $\gamma_u > \bar\gamma_{u_0}$,      $\mathcal{A}_u$  is Hurwitz stable. The $\bar\gamma_{u_0}$ is given by
\begin{equation}\label{eq gamma u0}
\begin{aligned}
\bar \gamma_{u_0}=\frac{\|P^{-1}_u(\infty)A_u+A^T_uP^{-1}_u(\infty) \|_2 + \tilde \gamma_{u_0}}{\lambda_{N-1}(\mathcal{L})},
\end{aligned}
\end{equation}
where $\tilde \gamma_{u_0} = 4N^2\alpha((P^{-1}_u(\infty)Q_uP^{-1}_u(\infty)+ C^T_{c,u}R^{-1}_{d,u}\\
C_{c,u})^{-1}) \alpha(C^T_{c,u}R^{-1}_{d,u}C_{c,u})$,  $\alpha(\cdot)$  is  the spectral abscissa as defined in the Notations,
 and $\gamma_u$ and $\mathcal{A}_u$ are  defined in (\ref{eq dot x_iu}) and (\ref{eq mathcal Au to s}), respectively.
\end{proposition}

Given Proposition~\ref{proposi gammau0}, select any  $\gamma_{u_0}$ satisfying  $\gamma_{u_0}> \bar  \gamma_{u_0}$ as the initial consensus parameter.  Define the parameter $\mathcal{A}_{u_0}$  as
\begin{equation}\label{eq Au0}
\begin{aligned}
\mathcal{A}_{u_0} =  G_{d,u} - \gamma_{u_0} (\mathcal{L}\otimes P_u(\infty)).
\end{aligned}
\end{equation}
It is worth mentioning that these assumptions are   primarily  provided  for  the   nominal  parameters   rather than the actual parameters.
Assumption~\ref{ass communication graph} is the basic assumption of the communication  networks.
Assumption~\ref{ass detectable uncontrolled}  provides   the controllability and observability  of the nominal model.
Proposition~\ref{proposi gammau0}   provides a   lower bound for  the consensus gain to  ensure that  the nominal matrix $\mathcal{A}_u$ is Hurwitz stable.



\subsection{Bounds of the Estimation Error Covariance}\label{sec sub bounds of the estimation error covariance}

This subsection    provides the bounds of the estimation error covariance by utilizing the nominal performance index and the parameter deviations,  and presents
 the lower bound  of  the nominal performance index.
Moreover, it reveals  the effect of the consensus parameter  on both bounds.

\begin{assumption}\label{ass Fd Fd}
If $F_{d,u}\neq 0$,  $A$ is a Hurwitz matrix. If $F_{d,u}= 0$, no
additional conditions apply.
\end{assumption}

\begin{remark}
For  Assumption~\ref{ass Fd Fd}, if $F_{d,u}\neq 0$ and $A$ is unstable, the estimation error covariance will diverge,  making it  meaningless to investigate the corresponding  result.
{\color{blue}
 Assumption  \ref{ass Fd Fd}  is introduced  to provide  a unified framework for analysis, covering both cases where  $F_{d,u}\neq 0$ and $F_{d,u}= 0$.
 The  case  $F_{d,u}=0$  primarily  reflects the scenarios  where   $A_u=A$ and $C_{i,u} = C_i$. For this case, the following  results  highlight the effect of  mismatched noise covariances on the filter's  performance.
 }
\end{remark}

\begin{proposition}\label{theorem Ahurwite sigmau sigmaa converge}
 Consider the system~(\ref{eq dynamics}) and the nominal distributed  filter~(\ref{eq dot x_iu}).  For  $\Sigma_u(t)$  and    $\Sigma_e(t)$,  it holds
\begin{enumerate}
\item
Under Assumptions~\ref{ass communication graph}-\ref{ass detectable uncontrolled}, with
$\gamma_u\geq \gamma_{u_0}$ defined in~(\ref{eq gamma u0}),
 $\Sigma_u(t)$   converges to $\bar \Sigma_u$,  where  $\bar \Sigma_u$ is    the solution of the following  Lyapunov equation
\begin{equation}\label{eq stead sigmau}
\begin{aligned}
 \mathcal{A}_u \bar \Sigma_{u} + \bar\Sigma_{u}\mathcal{A}^T_u + K_{d,u}R_{d,u} K^T_{d,u} +U_{N}\otimes Q_{u}=0.
\end{aligned}
\end{equation}
\item
Under Assumptions~\ref{ass communication graph}-\ref{ass Fd Fd}, with
$\gamma_u\geq \gamma_{u_0}$ defined in~(\ref{eq gamma u0}),
 $\Sigma_e(t)$  converges to $\bar \Sigma_e$, where  $\bar \Sigma_e$  is   the solution of the following  Lyapunov equation
\begin{equation}\label{eq stead sigmaa}
\begin{aligned}
& \mathcal{A}_u \bar\Sigma_e + \bar \Sigma_e\mathcal{A}^T_u + F_{d,u}\bar S^T + \bar S F^T_{d,u} \\
&~~~~~~~~~~+ K_{d,u}R_d K^T_{d,u} +U_{N}\otimes Q=0,
\end{aligned}
\end{equation}
where $\bar S$ and $\bar X$ are  the solutions of the following equations
$\mathcal{A}_u \bar S +  \bar S A^T_{d} + F_{d,u}\bar X +U_{N}\otimes Q = 0$
and $A_{d}\bar X+\bar XA^T_{d}+U_{N}\otimes Q=0$,  respectively.
%
\end{enumerate}

\end{proposition}

\begin{proof}
Item 1: Based on the results  in    Proposition~\ref{proposi gammau0}, it can be deduced  that   if  Assumption \ref{ass detectable uncontrolled}  holds, there exists $\gamma_{u_0}>0$  such that for all  $\gamma_u\geq \gamma_{u_0}$, $\mathcal{A}_u$  is Hurwitz stable.
Since $K_{d,u}$, $R_{d,u}$, and $Q_u$ are bounded,  it follows from Lemma~\ref{lemma sylvester equation} that
$\Sigma_u(t)$ will converge to  $\bar\Sigma_u$.

Item 2: In the scenario where  $F_{d,u} \neq 0$, the convergence   of
$S(t)$  affects  that of $\Sigma_e(t)$.
If $A$ is a  stable matrix,  it can be concluded that  $X(t)$, $S(t)$, and    $\Sigma_e(t)$ converge based on Lemma~\ref{lemma sylvester equation}.
  In the other case, if $F_{d,u}=0$, the convergence of  $\Sigma_e(t)$ can be derived based on the boundedness  of $Q$ and $R_d$, along with the principles of  the Lyapunov equation theory.
\end{proof}

%

%
%

%

%
%
%

For the  distributed filter with the nominal model, it is desired  that
the steady-state  estimation error covariance $\bar \Sigma_e$  can be evaluated by the nominal performance index $\bar\Sigma_u$. Consequently,  the bounds of the steady-state estimation error covariance are given  in the following.
Since   $\mathcal{A}_u$  incorporates  $\gamma_u$, and $\gamma_u$ may tend to infinity.  It is crucial  to identify the effect of $\gamma_u$.

\begin{theorem}\label{theorem bounds estimation error covariance}
Consider the system~(\ref{eq dynamics}) and the nominal distributed  filter~(\ref{eq dot x_iu}).
Suppose   Assumptions~\ref{ass communication graph}-\ref{ass Fd Fd} hold.  Define
$s = \|I_{nN} \otimes \mathcal{A}_u + A_{d,u}\otimes I_{nN}\|_F - \sqrt{N}\|\Delta A_d\|_F$  and  $\chi =\|I_{nN} \otimes A_{d,u} + A_{d,u}\otimes I_{nN}\|_F - 2\sqrt{nN}\| \Delta A_d\|_F$.
If $\gamma_u\geq \gamma_{u_0}$,    $s \neq 0$, and $\chi \neq 0$, then
\begin{enumerate}
 \item $\text{Tr} (\bar\Sigma_e)$ is bounded as follows
\begin{equation*}
\begin{aligned}
  \max\{0, \text{Tr} (\bar\Sigma_{u}) - \rho \}\leq\text{Tr} (\bar\Sigma_e) \leq  \text{Tr} (\bar\Sigma_{u}) + \rho,
\end{aligned}
\end{equation*}
where
\begin{equation}\label{eq theorem bounds rho}
\begin{aligned}
~~ \rho &= \Vert (\text{vec}(I_{nN}))^T \mathcal{\bar A}^{-1}_u (K_{d,u}\otimes K_{d,u}) \Vert_2 \\
 &~~~~\times \sqrt{\sum^N_{j=1}\Vert \Delta R_j \Vert^2_F }+ \Vert (\text{vec}(I_{nN}))^T \mathcal{\bar A}^{-1}_u \Vert_2\\
 &~~~~\times \Big(N\sqrt{\Vert\Delta Q\Vert^2_F} +  2 \|F_{d,u}\|_F\|\bar S\|_F\Big),
\end{aligned}
\end{equation}
%
\begin{equation}\label{eq bar s Fnorm}
\begin{aligned}
~~&~~~~\|\bar S\|_F\\
&\leq  \frac{\sqrt{nN}\|F_{d,u}\|_F \|\bar X\|_F + N\|Q_u\|_F+N\|\Delta Q\|_F}{\text{abs}(s)},~~
\end{aligned}
\end{equation}
and
\begin{equation}
\begin{aligned}
\|\bar X\|_F\leq \frac{N\|Q_u\|_F+N\|\Delta Q\|_F}{\text{abs}(\chi )}.~~~~~~~~~
\end{aligned}
\end{equation}


\item   $\rho$ is asymptotic decay  with respect to  $\gamma_u$, where
%
\begin{equation}\label{eq theorem gamma asymptotic}
\begin{aligned}
\|(\text{vec}(I_{nN}))^T \mathcal{\bar A}^{-1}_u \|_2 =\sqrt{a_1+\frac{b_1}{\gamma_u}+\frac{c_1}{\gamma^2_u}+O\big(\frac{1}{\gamma^3_u}\big)},~~~~~~~~
\end{aligned}
\end{equation}
\begin{equation}
\begin{aligned}
&~~~~\|(\text{vec}(I_{nN}))^T \mathcal{\bar A}^{-1}_u (K_{d,u}\otimes K_{d,u})\|_2 \\
&= \sqrt{a_2+\frac{b_2}{\gamma_u}+\frac{c_2}{\gamma^2_u}+O\big(\frac{1}{\gamma^3_u}\big)},~~~~~~~~~~~~~~~~~~
\end{aligned}
\end{equation}
and $a_1$, $b_1$, $c_1$,  $a_2$, $b_2$, and $c_2$ are constants, with   $a_1, b_1, a_2, b_2>0$.

%
%
%
%

\end{enumerate}

\end{theorem}

\begin{proof}
Item 1:
By defining  $\bar D = -F_{d,u}\bar S^T -\bar S F^T_{d,u}$,
(\ref{eq stead sigmaa}) can be rewritten as
\begin{equation}\label{eq barsigmaa rewritten}
\begin{aligned}
 &~~~~\mathcal{A}_u \bar\Sigma_e + \bar \Sigma_e\mathcal{A}^T_u + K_{d,u}R_{d,u} K^T_{d,u} +U_{N}\otimes Q_{u}\\
 & = K_{d,u}\Delta R_d K^T_{d,u} +U_{N}\otimes \Delta Q + \bar D.~~~~~~~~~~~~~~~~~~~~~~~~
\end{aligned}
\end{equation}
Now, consider the vectorization of  (\ref{eq barsigmaa rewritten}). According to  the identity $\text{vec}(XYZ)=(Z^T\otimes X)\text{vec}(Y)$,  one has
\begin{equation}\label{eq vec sigmaa}
\begin{aligned}
 &~~~~ \mathcal{\bar A}_u \text{vec} (\bar\Sigma_e)+ \text{vec}(K_{d,u}R_{d,u} K^T_{d,u} +U_{N}\otimes Q_{u})\\
 & = (K_{d,u}\otimes K_{d,u})\text{vec}(\Delta R_d) +\text{vec}(U_{N}\otimes \Delta Q)+\text{vec}(\bar D),
\end{aligned}
\end{equation}
where  $\mathcal{\bar A}_u = (I_{Nn}\otimes\mathcal{A}_u) +(\mathcal{A}_u\otimes I_{Nn})$.
Similarly to  (\ref{eq vec sigmaa}),  (\ref{eq stead sigmau}) can be  calculated as
\begin{equation}\label{eq vec sigmau}
\begin{aligned}
\mathcal{\bar A}_u \text{vec} (\bar\Sigma_{u}) + \text{vec}(K_{d,u}R_{d,u} K^T_{d,u} +U_{N}\otimes Q_{u})=0.
\end{aligned}
\end{equation}
By combining (\ref{eq vec sigmaa}) and (\ref{eq vec sigmau}),  it follows
\begin{equation*}
\begin{aligned}
 \mathcal{\bar A}_u \text{vec} (\bar\Sigma_e)&= \mathcal{\bar A}_u \text{vec} (\bar\Sigma_{u}) +  (K_{d,u}\otimes K_{d,u})\text{vec}(\Delta R_d) ~~~\\
 &~~~~+\text{vec}(U_{N}\otimes \Delta Q)+\text{vec}(\bar D).
\end{aligned}
\end{equation*}
 By utilizing the  properties of  the Kronecker sum and the fact that
 $\mathcal{A}_u$ is Hurwitz stable,  it can be deduced that $\mathcal{\bar A}_u$ is also Hurwitz stable. Hence, it holds
 \begin{equation}\label{vec sigmaa u sub}
\begin{aligned}
  \text{vec} (\bar\Sigma_e)&= \text{vec} (\bar\Sigma_{u}) +  \mathcal{\bar A}^{-1}_u (K_{d,u}\otimes K_{d,u})\text{vec}(\Delta R_d) \\
 &~~~~+\mathcal{\bar A}^{-1}_u\text{vec}(U_{N}\otimes \Delta Q)+\mathcal{\bar A}^{-1}_u\text{vec}(\bar D).
\end{aligned}
\end{equation}
By utilizing  the identity $\text{Tr}(X) = (\text{vec}(I))^T\text{vec}(X)$,  (\ref{vec sigmaa u sub}) can be rewritten as
\begin{equation}\label{eq trsigmaa-trsigmau}
\begin{aligned}
 \text{Tr} (\bar\Sigma_e)& = \text{Tr} (\bar\Sigma_{u}) + \rho_b,
\end{aligned}
\end{equation}
where
\begin{equation*}
\begin{aligned}
 \rho_b &=  (\text{vec}(I_{nN}))^T\mathcal{\bar A}^{-1}_u (K_{d,u}\otimes K_{d,u})\text{vec}(\Delta R_d) \\
 &~~~~+(\text{vec}(I_{nN}))^T\mathcal{\bar A}^{-1}_u\text{vec}(U_{N}\otimes \Delta Q) \\
 &~~~~+ (\text{vec}(I_{nN}))^T\mathcal{\bar A}^{-1}_u\text{vec}(\bar D).
\end{aligned}
\end{equation*}
Then, based on the  properties of (\ref{eq trsigmaa-trsigmau}),  one has
$\Vert\text{Tr} (\bar\Sigma_e) - \text{Tr} (\bar\Sigma_{u})\Vert_2 = \Vert\rho_b\Vert_2$.
 According to the identity $\Vert \text{vec}(X)\Vert_F = \Vert X\Vert_F$  and $\Vert X\Vert_2 \leq \Vert X\Vert_F$,   one has
\begin{equation}\label{eq bound for rho 1}
\begin{aligned}
 \Vert\rho_b\Vert_2 &\leq \Vert (\text{vec}(I_{nN}))^T \mathcal{\bar A}^{-1}_u (K_{d,u}\otimes K_{d,u})\Vert_2\\
 & ~~~~\times\sqrt{\sum^N_{j=1}\Vert \Delta R_j \Vert^2_F }+ \Vert (\text{vec}(I_{nN}))^T \mathcal{\bar A}^{-1}_u \Vert_2\\
 &~~~~\times \Big(N\sqrt{\Vert\Delta Q\Vert^2_F}
  + \Vert\bar D \Vert_F\Big).
\end{aligned}
\end{equation}
By defining   $\rho$   as  stated in (\ref{eq theorem bounds rho}),  it  can be deduced that $\Vert\text{Tr} (\bar\Sigma_e) - \text{Tr} (\bar\Sigma_{u})\Vert_2 \leq \rho$.


Next, we derive the upper bound of $\|\bar D\|_F$.
Based on  $\bar D = -F_{d,u}\bar S^T -\bar S F^T_{d,u}$,  we have
\begin{equation}\label{eq bound for rho 2}
\begin{aligned}
\|\bar D\|_F \leq 2 \|F_{d,u}\|_F\|\bar S\|_F.
\end{aligned}
\end{equation}
According to the expression  of $\bar S$ below (\ref{eq stead sigmaa}) and applying  the vectorization operation to it, it follows that
\begin{equation}\label{eq vec bar S}
\begin{aligned}
&~~~~(I_{nN} \otimes \mathcal{A}_u + A_{d,u}\otimes I_{nN})\text{vec}(\bar S)\\
&=(\Delta A_d \otimes I_{nN})\text{vec}(\bar S)
-(I_{nN}\otimes F_{d,u})\text{vec}(\bar X)\\
&~~~~-\text{vec}(U_N\otimes Q).
\end{aligned}
\end{equation}
By taking the Frobenius norm of (\ref{eq vec bar S}),  if $\|I_{nN} \otimes \mathcal{A}_u + A_{d,u}\otimes I_{nN}\|_F -\sqrt{nN}\|\Delta A_d\|_F>0$,  then
\begin{equation}\label{eq inequ s 1}
\begin{aligned}
& ~~~~\|I_{nN} \otimes \mathcal{A}_u + A_{d,u}\otimes I_{nN}\|_F\|\bar S\|_F\\
&\leq  \sqrt{nN}\|\Delta A_d\|_F\|\bar S\|_F
+ \sqrt{nN}\|F_{d,u}\|_F \|\bar X\|_F \\
&~~~~+ \|U_N\otimes Q\|_F.
\end{aligned}
\end{equation}
If $\|I_{nN} \otimes \mathcal{A}_u + A_{d,u}\otimes I_{nN}\|_F -\sqrt{nN}\|\Delta A_d\|_F<0$,  then
\begin{equation}\label{eq inequ s 2}
\begin{aligned}
&\sqrt{nN}\|\Delta A_d\|_F\|\bar S\|_F\leq   \|I_{nN} \otimes \mathcal{A}_u + A_{d,u}\otimes I_{nN}\|_F\|\bar S\|_F\\
&~~~~~~~~~~+ \sqrt{nN}\|F_{d,u}\|_F \|\bar X\|_F + \|U_N\otimes Q\|_F.
\end{aligned}
\end{equation}
 By combining  (\ref{eq inequ s 1}) and (\ref{eq inequ s 2}),
it can be concluded that  if  $\|I_{nN} \otimes \mathcal{A}_u + A_{d,u}\otimes I_{nN}\|_F -\sqrt{nN}\|\Delta A_d\|_F \neq 0$,  then  $\|\bar S\|_F$ can be bounded  as
\begin{equation}\label{eq bound for rho 3}
\begin{aligned}
&~~~~\|\bar S\|_F\\
&\leq  \frac{\sqrt{nN}\|F_{d,u}\|_F \|\bar X\|_F + N\|Q_u\|_F+N\|\Delta Q\|_F}{\text{abs}(\|I_{nN} \otimes \mathcal{A}_u + A_{d,u}\otimes I_{nN}\|_F - \sqrt{nN}\|\Delta A_d\|_F)}.
\end{aligned}
\end{equation}
Similarly, based on  $A_{d}\bar X+\bar XA^T_{d}+U_{N}\otimes Q=0$, we can obtain that if  $\|I_{nN} \otimes A_{d,u} + A_{d,u}\otimes I_{nN}\|_F - 2\sqrt{nN}\| \Delta A_d\|_F \neq 0$, then
\begin{equation}\label{eq bound for rho 4}
\begin{aligned}
&~~~~\|\bar X\|_F\\
&\leq \frac{N\|Q_u\|_F+N\|\Delta Q\|_F}{\text{abs}(\|I_{nN} \otimes A_{d,u} + A_{d,u}\otimes I_{nN}\|_F - 2\sqrt{nN}\| \Delta A_d\|_F )}.
\end{aligned}
\end{equation}
Based on (\ref{eq bound for rho 1}), (\ref{eq bound for rho 2}),
(\ref{eq bound for rho 3}), and (\ref{eq bound for rho 4}), the results can be proven.

Item 2:  First, we transform the problem of computing  $(\text{vec}(I_{nN}))^T \mathcal{\bar A}^{-1}_u$   into finding  the solution $X$ of  the Lyapunov equation $\mathcal{A}^T_uX+X\mathcal{A}_u = I_{nN}$.
Since  $(\text{vec}(I_{nN}))^T \mathcal{\bar A}^{-1}_u   =  (\mathcal{\bar A}^{-T}_u \text{vec}(I_{nN}))^T$,  our   focus is on  $\mathcal{\bar A}^{-T}_u \text{vec}(I_{nN})$.
 Given that $\lambda(\mathcal{A}_u) \cap \lambda(-\mathcal{A}^T_u) = \emptyset$,   $X$ is  the unique  solution of the  Lyapunov equation   $\mathcal{A}^T_uX+X\mathcal{A}_u = I_{nN}$. Thus,
   we have $\text{vec}(X)=\mathcal{\bar A}^{-T}_u \text{vec}(I_{nN})$,  where
$\mathcal{\bar A}^T_u = (I_{nN}\otimes\mathcal{A}^T_u) +(\mathcal{A}^T_u\otimes I_{nN})$.

Second, we  derive   an analytic  solution  for  $X$   in terms of~$\gamma_u$. Based on the  property of the  communication topology~$\mathcal{L}$, there exist a transformation matrix  $\bar T$ such that $\tilde \Lambda = \bar T \mathcal{L}\bar T^T = \text{diag}\{0, \lambda_2,\ldots,\lambda_N\}$, where $\lambda_2,\ldots,\lambda_N$ are the positive eigenvalues of $\mathcal{L}$, and  the first row  of $\bar T$ is the  vector with all elements equal to  $1$.
Denote $T= \bar T \otimes I_n$,  then $\gamma_u T(\mathcal{L}\otimes P_u(\infty))T^T = \gamma_u (\tilde \Lambda\otimes P_u(\infty))$.
 We can rewrite  the  Lyapunov equation as
\begin{equation}\label{eq tilde G Lyapunov}
\begin{aligned}
&(\tilde G^T - \gamma_u (\tilde \Lambda\otimes P_u(\infty)))\tilde X \\
&+ \tilde X (\tilde G - \gamma_u (\tilde \Lambda\otimes P_u(\infty)))= I_{nN},
\end{aligned}
\end{equation}
where $\tilde G = T G_{d,u}T^T$ and $\tilde X = TXT^T$.
Split   the matrix  $\tilde G$ into  blocks as follows
\begin{equation}\label{eq block tilde G}
\begin{aligned}
\tilde G = \left[\begin{array}{cc}
\tilde G_{11} & \tilde G_{12}\\
\tilde G^T_{12} & \tilde G_{22}
\end{array}\right],
\end{aligned}
\end{equation}
where $\tilde G_{11} \in \mathcal{R}^{n\times n}$, $\tilde G_{12} \in \mathcal{R}^{n\times n(N-1)}$, and $\tilde G_{22} \in \mathcal{R}^{n(N-1)\times n(N-1)}$. Denote $\Lambda=  \text{diag}\{\lambda_2,\ldots,\lambda_N\}$, and  we have
\begin{equation}\label{eq tilde G minus gammau}
\begin{aligned}
(\tilde G - \gamma_u (\tilde \Lambda\otimes P_u(\infty)))= \left[\begin{array}{cc}
\tilde G_{11} & \tilde G_{12}\\
\tilde G^T_{12} & \tilde G_{22}-\gamma_u \Lambda\otimes P_u(\infty)
\end{array}\right].
\end{aligned}
\end{equation}
We can  define  the matrix  $\tilde X$ as   a block matrix with the same  dimensions as the sub-blocks  of $\tilde G$ in   (\ref{eq block tilde G}). By examining   these block matrices,  it can be  hypothesized that
the solution of $\tilde X$  is  analytic in $\gamma_u$. The  analytic form can be determined by  utilizing the block matrices  in  (\ref{eq tilde G Lyapunov}) and  analyzing four  equalities.  The solution of $\tilde X (\gamma_u)$ can be given as
\begin{equation}\label{eq tilde X block}
\begin{aligned}
\tilde X(\gamma_u) = \left[\begin{array}{cc}
\tilde X_{11,1} +\frac{1}{\gamma_u}\tilde X_{11,2} +O(\frac{1}{\gamma^2_u})  & \frac{1}{\gamma_u}\tilde X_{12} +O(\frac{1}{\gamma^2_u}) \\
\frac{1}{\gamma_u}\tilde X^T_{12} +O(\frac{1}{\gamma^2_u}) & \frac{1}{\gamma_u}\tilde X_{22} +O(\frac{1}{\gamma^2_u})
\end{array}\right],
\end{aligned}
\end{equation}
where   $\tilde X_{22}$ is the unique solution of
$(\Lambda\otimes P_u(\infty))\tilde X_{22}+\tilde X_{22} (\Lambda\otimes P_u(\infty)) = I_{n(N-1)}$,  $\tilde X_{11,1}$ is the unique solution of $\tilde G_{11}\tilde X_{11,1}+ \tilde X_{11,1} \tilde G_{11} = I_n$,  $\tilde X_{12}$ is  $\tilde X_{12} =  \tilde X_{11,1}\tilde G_{12} (\Lambda\otimes P_u(\infty))^{-1}$,  and $\tilde X_{11,2}$ is the unique  solution of  $\tilde G_{11}\tilde X_{11,2}+  \tilde X_{11,2}\tilde G_{11} + \tilde G_{12}\tilde X^T_{12} + \tilde X_{12}\tilde G^T_{12} = 0$.

Finally, since $X(\gamma_u)= T^T\tilde X(\gamma_u)T$,
 multiplying a constant matrix into $\tilde X(\gamma_u)$  does not influence the structure of  each term of $X(\gamma_u)$ in terms of $\gamma_u$.  We can conclude
\begin{equation}\label{eq vec X norm a b c}
\begin{aligned}
\|\text{vec}(X(\gamma_u))\|_2 = \sqrt{a_1+\frac{b_1}{\gamma_u}+\frac{c_1}{\gamma^2_u}+O\big(\frac{1}{\gamma^3_u}\big)},
\end{aligned}
\end{equation}
where $a_1$, $b_1$, and $c_1$ are  constants with $a_1, b_1>0$.
Similarly,  $\Vert (\text{vec}(I_{nN}))^T \mathcal{\bar A}^{-1}_u (K_{d,u}\otimes K_{d,u})\Vert_2$ has a similar form as~(\ref{eq vec X norm a b c}).


%
%



\end{proof}

%

%
%

\begin{remark}
The inequality condition in  Theorem~\ref{theorem bounds estimation error covariance}  can be easily satisfied by slightly increasing  the deviation  $\| \Delta A_d\|_F$. Additionally,    $\bar X$ can be evaluated  and replaced by using its known upper bound.
Theorem \ref{theorem bounds estimation error covariance}   provides   bounds of   the steady-state  performance of the distributed  continuous-time filter. These bounds  can be evaluated utilizing the  nominal performance index and the  knowledge of the Frobenius norms of $\Delta A$, $\Delta C_i$, $\Delta Q$, and $\Delta R_i$.
Item~2 reveals the effect of the consensus parameter $\gamma_u$ on the upper bound of the estimation error covariance, and provides  the asymptotic  decay  behavior of this bound   in terms of $\gamma_u$.

%
%
%
%

\end{remark}

%
%

Furthermore, it is desired to establish   the lower bound of  $\text{Tr}(\bar\Sigma_{u})$,  which  can also   be  utilized to evaluate the estimation error covariance.
 The lower  bound of  $\text{Tr}(\bar\Sigma_{u})$ is  provided   as follows.

\begin{proposition}\label{proposition lower bound trsigmau}
 Consider the system~(\ref{eq dynamics}) and the nominal distributed  filter~(\ref{eq dot x_iu}).
Suppose   Assumptions~\ref{ass communication graph}-\ref{ass Fd Fd} hold. If $\gamma_u\geq \gamma_{u_0}$,  then
  the lower bound of $\text{Tr}(\bar\Sigma_{u})$  is
\begin{equation}\label{eq lower bound of barsigmau}
\begin{aligned}
\text{Tr}(\bar\Sigma_{u})&\geq\frac{\text{Tr}(K_{d,u}R_{d,u} K^T_{d,u} +U_{N}\otimes Q_{u})}{2\text{Tr}(-\mathcal{A}_{u_0}) + 2 (\gamma_{u}-\gamma_{u_0}  )\text{Tr}( \mathcal{L})\text{Tr}(P_u(\infty))}.
\end{aligned}
\end{equation}
\end{proposition}

\begin{proof}
Taking the trace operation  on  (\ref{eq stead sigmau}),  one has
\begin{equation}\label{eq theo bar sigmau bound}
\begin{aligned}
 &~~~~\text{Tr}(\mathcal{A}_u \bar \Sigma_{u}) + \text{Tr}(\bar\Sigma_{u}\mathcal{A}^T_u)\\
 & =-\text{Tr}(K_{d,u}R_{d,u} K^T_{d,u}) -\text{Tr}(U_{N}\otimes Q_{u}).
\end{aligned}
\end{equation}
Based on the property of the trace $\text{Tr}(XY) = \text{Tr}(YX)$, it follows
\begin{equation}\label{eq Au=sigmasigmaAu}
\begin{aligned}
\text{Tr}(\mathcal{A}^T_u ) &= \text{Tr}(\bar\Sigma^{-1}_{u}\bar\Sigma_{u}\mathcal{A}^T_u) \\ &=\text{Tr}(\bar\Sigma_{u}\mathcal{A}^T_u\bar\Sigma^{-1}_{u})\\
&=\text{Tr}(\bar\Sigma^{-1}_{u}\mathcal{A}_u\bar\Sigma_{u}).
\end{aligned}
\end{equation}
According to   (\ref{eq Au=sigmasigmaAu}) and  (\ref{eq stead sigmau}), it holds
\begin{equation}\label{eq Au=sigmasigmaAu 2}
\begin{aligned}
2\text{Tr}(\mathcal{A}^T_u )& = \text{Tr}(\bar\Sigma^{-1}_{u}(\bar\Sigma_{u}\mathcal{A}^T_u+\mathcal{A}_u\bar\Sigma_{u}))\\
& = - \text{Tr}(\bar\Sigma^{-1}_{u}(K_{d,u}R_{d,u} K^T_{d,u} +U_{N}\otimes Q_{u}) ).
\end{aligned}
\end{equation}
By utilizing  Lemma  \ref{lemma tr AB trA B}  and  the fact that   $K_{d,u}R_{d,u} K^T_{d,u} +U_{N}\otimes Q_{u}>0$ and $\bar\Sigma_{u}>0$,  one has
\begin{equation}\label{eq Au=sigmasigmaAu 3}
\begin{aligned}
&~~~~\text{Tr}(\bar\Sigma^{-1}_{u}(K_{d,u}R_{d,u} K^T_{d,u} +U_{N}\otimes Q_{u}) )\\
 &\geq  (\text{Tr}(\bar\Sigma_{u}))^{-1}\text{Tr}(K_{d,u}R_{d,u} K^T_{d,u} +U_{N}\otimes Q_{u}).~~~~~
\end{aligned}
\end{equation}
By combining  (\ref{eq Au=sigmasigmaAu 2})  and (\ref{eq Au=sigmasigmaAu 3}),  it follows
\begin{equation}\label{eq Au=sigmasigmaAu 44}
\begin{aligned}
~~2\text{Tr}(-\mathcal{A}_u )\text{Tr}(\bar\Sigma_{u})& \geq \text{Tr}(K_{d,u}R_{d,u} K^T_{d,u} +U_{N}\otimes Q_{u}).
\end{aligned}
\end{equation}

Since   all eigenvalues of  $\mathcal{A}_u$ are negative,  (\ref{eq Au=sigmasigmaAu 44}) can be  rewritten as
\begin{equation}\label{eq Au=sigmasigmaAu 4}
\begin{aligned}
\text{Tr}(\bar\Sigma_{u})& \geq \frac{1}{2} (\text{Tr}(-\mathcal{A}_u ))^{-1}\text{Tr}(K_{d,u}R_{d,u} K^T_{d,u} +U_{N}\otimes Q_{u}).
\end{aligned}
\end{equation}
Note that  there exists
\begin{equation}\label{eq aut au -2gamma}
\begin{aligned}
-\mathcal{A}_{u} =  -\mathcal{A}_{u_0}  + (\gamma_{u}-\gamma_{u_0}  ) (\mathcal{L}\otimes P_u(\infty)).
\end{aligned}
\end{equation}
Based on    (\ref{eq aut au -2gamma}),   it holds $\text{Tr}(-\mathcal{A}_u ) =  \text{Tr}(-\mathcal{A}_{u_0}) +  (\gamma_{u}-\gamma_{u_0}  )\text{Tr}( \mathcal{L})\text{Tr}(P_u(\infty))$.
Therefore,  (\ref{eq Au=sigmasigmaAu 4}) can be simplified as
\begin{equation}\label{eq theo lower boudnu 11}
\begin{aligned}
\text{Tr}(\bar\Sigma_{u})&\geq\frac{\text{Tr}(K_{d,u}R_{d,u} K^T_{d,u} +U_{N}\otimes Q_{u})}{2\text{Tr}(-\mathcal{A}_{u_0}) + 2 (\gamma_{u}-\gamma_{u_0}  )\text{Tr}( \mathcal{L})\text{Tr}(P_u(\infty))}.
\end{aligned}
\end{equation}
\end{proof}

%
%

\begin{remark}
By combining  Proposition~\ref{proposition lower bound trsigmau} and  Theorem~\ref{theorem bounds estimation error covariance},
the lower bound of
$\text{Tr} (\bar\Sigma_e)$ can be directly  evaluated without the need for computing  $\text{Tr} (\bar\Sigma_{u})$.  The relations between the lower bound of $\text{Tr} (\bar\Sigma_{u})$ and the model parameters  are  presented in Proposition~\ref{proposition lower bound trsigmau}, and  the consensus parameter is separated  in the lower bound. It is shown that as the consensus parameter tends to infinity, the lower bound  converges  to 0.

\end{remark}

\subsection{Divergence Analysis}\label{sec divergence analysis}

This subsection  considers the divergence of the  distributed continuous-time filter   due to the incorrect  process noise covariance.

\begin{theorem}\label{theorem incorrect Q divergence}
 Consider the system~(\ref{eq dynamics}) and the nominal distributed  filter~(\ref{eq dot x_iu}).
 Suppose   Assumption~\ref{ass communication graph} hold.
Let  $r$  be  a real number,  $j$ be the imaginary unit, and $e$  be  a vector with the  appropriate dimension.
If $A^T_ue =rje$ and $Q_ue=0$,  then
\begin{enumerate}
\item  $rj$  is   an eigenvalue of  $\mathcal{A}^T_u$ with the  eigenvector   $1_N\otimes e$, i.e.,  $\mathcal{A}^T_u(1_N\otimes e)  = rj(1_N\otimes e)$, for any consensus parameter $\gamma_u$ and any undirected connected  topology~$\mathcal{L}$.

\item
 If $\Delta F_{d,u}=0$ and  $Qe\neq0$, it holds $\lim_{t\to\infty}  (1^T_N\otimes e^*) \Sigma_e(t)(1_N\otimes e)  \to \infty.$

\end{enumerate}

\end{theorem}

\begin{proof}
Item 1: If $A^T_ue =rje$ and $Q_ue=0$,    one has $e^*A_u =-rje^*$.
By  pre-multiplying  with $e^*$ and post-multiplying   with $e$ in (\ref{eq steady Pu Lyapunove euqaiton}), the subsequent expression is derived
\begin{equation}\label{eq lemma elye}
\begin{aligned}
 ~~~~~&e^*A_u P_u(\infty)e + e^*P_u(\infty)A^T_ue + e^*Q_ue\\
 &- e^*P_u(\infty)C^T_{c,u}R^{-1}_{d,u}C_{c,u}P_u(\infty)e=0.
\end{aligned}
\end{equation}
Based on   $A^T_ue =rje$, $e^*A_u =-rje^*$,  and   $Q_ue=0$,  (\ref{eq lemma elye})  can be simplified as
 \begin{equation}
\begin{aligned}
 &~~~~e^*P_u(\infty)C^T_{c,u}R^{-1}_{d,u}C_{c,u}P_u(\infty)e\\
 &=  \sum^N_{j=1} e^*P_u(\infty)C^T_{j,u}R^{-1}_{j,u}C_{j,u}P_u(\infty)e~~~~~\\
 &= 0.
\end{aligned}
\end{equation}
 Since  $e^*P_u(\infty)C^T_{j,u}R^{-1}_{j,u}C_{j,u}P_u(\infty)e \geq 0$,   it can be concluded that
 \begin{equation}\label{eq RCPe}
\begin{aligned}
R^{-1}_{i,u}C_{i,u}P_u(\infty)e = 0, i=1, \ldots, N.
\end{aligned}
\end{equation}
 Utilizing (\ref{eq mathcal Au}) and the property of the undirected connected  graph,  post-multiply $\mathcal{A}^T_u$ by the vector $1_N\otimes e$, yielding
\begin{equation}
\begin{aligned}
\mathcal{A}^T_u(1_N\otimes e) &=  G^T_{d,u} (1_N\otimes e)- \gamma_u (\mathcal{L}^T\otimes P_u(\infty))(1_N\otimes e)\\
&= [\bar e^T_1,\ldots,\bar e^T_N]^T,
\end{aligned}
\end{equation}
where $ \bar e_i =  A^T_u e- N C^T_{i,u}R^{-1}_{i,u}C_{i,u}P_u(\infty)e$.  Exploiting  the fact $A^T_ue=rje$ and (\ref{eq RCPe}),  it is evident that
 $\mathcal{A}^T_u(1_N\otimes e)  = rj(1_N\otimes e)$.  Consequently,  $rj$  is  a  purely imaginary  eigenvalue of  $\mathcal{A}_u$.

Item 2: Item 1 and (\ref{eq RCPe}) lead to the conclusion that  $R^{-1}_{i,u}C_{i,u}P_u(\infty)e = 0$.  Considering the expression   $K_{d,u} = \text{diag}\{K_{1,u},\ldots,K_{N,u}\}$,  it can be inferred that  $K^T_{d,u}(1_N\otimes e) = 0$. Additionally,
Item~1  demonstrates   that $\mathcal{A}^T_u(1_N\otimes e)  = rj(1_N\otimes e)$. By pre-multiplying  with $1^T_N\otimes e^*$, and post-multiplying   with $1_N\otimes e$ in  (\ref{eq sigmaa fdu=0}),  one has
\begin{equation}\label{eq contradiction prove sigmaa}
\begin{aligned}
&~~~~(1^T_N\otimes e^*)\dot \Sigma_e(t)(1_N\otimes e)\\
& = (1^T_N\otimes e^*)\mathcal{A}_u \Sigma_e(t)(1_N\otimes e)\\
&~~~~+  (1^T_N\otimes e^*)\Sigma_e(t)\mathcal{A}^T_u(1_N\otimes e)\\
&~~~~+ (1^T_N\otimes e^*)K_{d,u}R_d K^T_{d,u}(1_N\otimes e) \\ &~~~~+(1^T_N\otimes e^*)U_{N}\otimes Q(1_N\otimes e)\\
&\geq (1^T_N\otimes e^*)U_{N}\otimes Q(1_N\otimes e)>0.
\end{aligned}
\end{equation}
It can be observed that  (\ref{eq contradiction prove sigmaa})
as $t$ tends to infinity, $(1^T_N\otimes e^*) \Sigma_e(t)(1_N\otimes e)$ will  continuously  increase.
Consequently, this term  diverges, and the proof is concluded.
\end{proof}
%
%

\begin{remark}
Item~1 highlights  that     if  $A^T_u$  possesses   an eigenvalue $rj$ on the imaginary axis, which is linked to the right eigenvector $e$ in the  right null space of $Q_u$, then $\mathcal{A}^T_u$  will also have  $rj$ as its eigenvalue with the right eigenvector $1_N\otimes e$.
Item~2 reveals   that an incorrect process noise covariance can  lead to  the divergence of the estimation error covariance in the distributed continuous-time filter,  even though the nominal performance index may converge.
  Therefore,  it is  advisable  to  ensure that   the nominal model can be designed  in such a way that
   Assumption  \ref{ass detectable uncontrolled}  holds, thus avoiding  the presence of  pure imaginary eigenvalues. The effect of  the consensus parameter on our conclusion  is  eliminated.
 It is important to note that  the prerequisite for  the communication graph to be undirected  must be fulfilled. Otherwise, this inference may not be valid  when assuming  a directed graph.
\end{remark}

\begin{remark}
When the  communication topology  is    mismatched, it  can lead to  the divergence of the estimation error covariance.  This  may  violate   the condition  presented in Proposition~\ref{proposi gammau0},  resulting in  the  instability  of  $\mathcal{A}_u$. When the mismatched topology is denoted as  $\mathcal{L}_u$, {\color{blue}
we can  use     the lower bound of  $\lambda_{N-1}(\mathcal{L}_u)$  to compute $\bar \gamma_{u_0}$ in (\ref{eq gamma u0}), ensuring  that   $\mathcal{A}_u$ is Hurwitz stable by choosing $\gamma_u > \bar \gamma_{u_0}$.
}
\end{remark}

\subsection{Relations between $\Sigma_u(t)$ and $\Sigma_e(t)$}

Section \ref{sec sub bounds of the estimation error covariance}  shows that
the bounds of  $\text{Tr} (\bar\Sigma_e)$ can be evaluated by
$\text{Tr} (\bar\Sigma_{u})$ and the  Frobenius norms  of the  parameter  deviation information.   Furthermore,
this subsection  explores  the  relation between  $\Sigma_u(t)$ and $\Sigma_e(t)$    by relying    on the   information about the magnitude relation between  $R_{i,u}$ and  $R_i$   as well as  between  $Q_u$ and $Q$.   It is assumed that  $\Delta A=0$ and $\Delta C_{i,u}=0, i=1,\ldots, N$.

First, define the difference $E_{ue}(t)$ between the nominal  performance index  $\Sigma_u(t)$  and the estimation error covariance $\Sigma_e(t)$  as  $E_{ue}(t) = \Sigma_u(t)-\Sigma_e(t).$
Then,  by utilizing  (\ref{eq sigmau dot time}) and (\ref{eq sigmaa fdu=0}),
the  differential equation of  $\dot E_{ue}(t)$  can be computed as $\dot E_{ue}(t)  = \mathcal{A}_u E_{ue}(t) + E_{ue}(t)\mathcal{A}^T_u+ K_{d,u}\Delta R_{d} K^T_{d,u} +U_{N}\otimes \Delta Q.$
Define the time index $h$ satisfying  $t\geq h \geq 0$.
Drawing upon the principles of the  linear system theory,  a precise  analytical solution   can be calculated as
\begin{equation}\label{eq eua expression}
\begin{aligned}
E_{ue}(t) & = e^{\mathcal{A}_u(t-h)}E_{ue}(h)e^{\mathcal{A}^T_u(t-h)}+ \int_h^t  e^{\mathcal{A}_u(t-\tau)}\\
&~~~~\times ( K_{d,u}\Delta R_{d} K^T_{d,u} +U_{N}\otimes \Delta Q)     e^{\mathcal{A}^T_u(t-\tau)} d\tau.
\end{aligned}
\end{equation}
For further analysis,  define
 \begin{equation}\label{eq delta D}
\begin{aligned}
\Delta D = K_{d,u}\Delta R_{d} K^T_{d,u} +U_{N}\otimes \Delta Q.
\end{aligned}
\end{equation}
 The subsequent theorem  shows  the relations between    $\Sigma_u(t)$ and $\Sigma_e(t)$.    Moreover,  it   demonstrates  how  the upper bound of the  spectral  norm  of the  difference between   $\Sigma_u(t)$ and $\Sigma_e(t)$ changes  as  time progresses.

\begin{theorem}\label{theorem relation sigmau sigmaa}
 Consider the system~(\ref{eq dynamics}) and the nominal distributed  filter~(\ref{eq dot x_iu}).
 Suppose   Assumptions~\ref{ass communication graph}-\ref{ass detectable uncontrolled} hold.
If
$\gamma_u\geq \gamma_{u_0}$, $\Delta A=0$, and $\Delta C_{i,u}=0, i=1,\ldots, N$,  it holds
\begin{enumerate}
\item   If   $\Delta D \geq 0$ and  $E_{ue}(h)\geq 0$, then  $\Sigma_u(t)\geq \Sigma_e(t)$.
If   $\Delta D \leq 0$ and  $E_{ue}(h)\leq 0$, then $\Sigma_u(t)\leq \Sigma_e(t)$.

%


\item \label{theorem relation item e bounds}
Define $ \bar \mu(\mathcal{A}_{u_0}) =\mu(\mathcal{A}_{u_0})+\mu(\mathcal{A}^T_{u_0})$ and $ \mathcal{L}_{\gamma} = - (\gamma_{u}-\gamma_{u_0}  ) (\mathcal{L}\otimes P_u(\infty))$. Then,
$\Vert E_{ue}(t) \Vert_2$ has the following   bound
\begin{equation}\label{eq eua(t) bound}
\begin{aligned}
\Vert E_{ue}(t) \Vert_2 &\leq  \Vert E_{ue}(h)\Vert_2  e^{\bar \mu(\mathcal{A}_{u_0})(t-h)}e^{2\mu(\mathcal{L}_{\gamma})(t-h)}\\
& ~~+  \Vert\Delta D \Vert_2   \int_h^t   e^{\bar \mu(\mathcal{A}_{u_0})(t-\tau)}e^{2\mu(\mathcal{L}_{\gamma})(t-\tau)} d\tau.
\end{aligned}
\end{equation}
 Specifically,  when the logarithmic norm is chosen as $\mu_2(\cdot)$, one has $\mu_2(\mathcal{L}_{\gamma}) = 0$.

%

%

\end{enumerate}

\end{theorem}

\begin{proof}
Item 1):   By combining  (\ref{eq eua expression}), these  conclusions can be easily  drawn.

Item 2): Taking the spectral norm of  (\ref{eq eua expression}),  it follows
\begin{equation}\label{eq lemma compare bound 1}
\begin{aligned}
\Vert E_{ue}(t) \Vert_2 &\leq  \Vert E_{ue}(h)\Vert_2 \Vert e^{\mathcal{A}_u(t-h)} \Vert_2\Vert e^{\mathcal{A}^T_u(t-h)} \Vert_2\\
& ~~~~+  \Vert\Delta D \Vert_2   \int_h^t  \Vert e^{\mathcal{A}_u(t-\tau)} \Vert_2\Vert e^{\mathcal{A}^T_u(t-\tau)} \Vert_2 d\tau.
\end{aligned}
\end{equation}
 Based on Lemma  \ref{lemma mu logarithmic}, one has
\begin{equation}\label{eq lemma compare bound 2}
\begin{aligned}
 \Vert e^{\mathcal{A}_u(t-\tau)} \Vert_2\Vert e^{\mathcal{A}^T_u(t-\tau)} \Vert_2 &\leq  e^{\mu(\mathcal{A}_u)(t-\tau)}e^{\mu(\mathcal{A}^T_u)(t-\tau)}\\
 & \leq e^{(\mu(\mathcal{A}_u)+\mu(\mathcal{A}^T_u))(t-\tau)}.
\end{aligned}
\end{equation}
Similarly to   Theorem \ref{theorem bounds estimation error covariance},   based on the definition that  $\mathcal{A}_{u_0} =  G_{d,u} - \gamma_{u_0} (\mathcal{L}\otimes P_u(\infty))$, $ \mathcal{L}_{\gamma} = - (\gamma_{u}-\gamma_{u_0}  ) (\mathcal{L}\otimes P_u(\infty))$, and $\mathcal{A}_{u} =  \mathcal{A}_{u_0} + \mathcal{L}_{\gamma}$, $ e^{(\mu(\mathcal{A}_u)+\mu(\mathcal{A}^T_u))(t-\tau)}$ can be rewritten as
\begin{equation}\label{eq lemma compare bound 3}
\begin{aligned}
  e^{(\mu(\mathcal{A}_u)+\mu(\mathcal{A}^T_u))(t-\tau)} =  e^{(\mu(\mathcal{A}_{u_0})+\mu(\mathcal{A}^T_{u_0}))(t-\tau)}e^{2\mu(\mathcal{L}_{\gamma})(t-\tau)}.
\end{aligned}
\end{equation}
By combining (\ref{eq lemma compare bound 2}) and (\ref{eq lemma compare bound 3}), (\ref{eq lemma compare bound 1}) can be rewritten as
\begin{equation}
\begin{aligned}
\Vert E_{ue}(t) \Vert_2 &\leq  \Vert E_{ue}(h)\Vert_2  e^{(\mu(\mathcal{A}_{u_0})+\mu(\mathcal{A}^T_{u_0}))(t-h)}e^{2\mu(\mathcal{L}_{\gamma})(t-h)}\\
& ~~~~+  \Vert\Delta D \Vert_2   \int_h^t   e^{(\mu(\mathcal{A}_{u_0})+\mu(\mathcal{A}^T_{u_0}))(t-\tau)}\\
&~~~~\times e^{2\mu(\mathcal{L}_{\gamma})(t-\tau)} d\tau.
\end{aligned}
\end{equation}
Define $\bar \mu(\mathcal{A}_{u_0}) =\mu(\mathcal{A}_{u_0})+\mu(\mathcal{A}^T_{u_0})$, and (\ref{eq eua(t) bound}) can be obtained.
Based on Assumption \ref{ass communication graph}, the  Laplacian  matrix $\mathcal{L}$   is positive semi-definite with one  zero eigenvalue.
According to the definition of  $\mu_2(\cdot)$ and the property of   $\mathcal{L}_{\gamma}$, it can be deduced  that
$\mu_2(\mathcal{L}_{\gamma}) = 0$.

\end{proof}

%
%

\begin{remark}
On the one hand, Theorem  \ref{theorem relation sigmau sigmaa}  shows that the  relative relation between    $\Sigma_u(t)$ and $\Sigma_e(t)$
can be determined by  utilizing  $\Delta D$ and $E_{ue}(h)$.
On the other hand,  the evolution process of
the upper bound of the spectral norm of the difference is  presented.
Item~2 provides a distinct separation of
  the consensus parameter in  the upper bound.
  The  logarithmic norm
 $\mu_2(\mathcal{L}_{\gamma})=0$ means that  the magnitude of the coefficient of the exponential term is influenced by  the logarithmic norm of   $\mathcal{A}_{u_0}$ associated with   $\gamma_{u_0}$, and  the  consensus parameter   $\gamma_u$ satisfying  $\gamma_u\geq \gamma_{u_0}$   have no  impact on the bound. If $\Delta D = 0$,  Item~\ref{theorem relation item e bounds} show the convergence rate between the nominal performance index $\Sigma_u(t)$ and the estimation error covariance $\Sigma_e(t)$.
\end{remark}

\section{Simulations}\label{sec simulations}


{\color{blue}
This section conducts  numerical experiments  on  vehicle   tracking to  validate the effectiveness of the theoretical results  regarding  the effects of modeling errors on the filter performance.
Consider the vehicle  moving on a  plane. In engineering applications, we primary focus on  the vehicle's   positions and  velocities. Let
$x(t) = [x^T_1(t), x^T_2(t), x^T_3(t), x^T_4(t)]^T$  denote  the state,  where $x_1(t)$ and $x_3(t)$  represent   the  horizontal    and   vertical   velocities, respectively, and $x_2(t)$ and $x_4(t)$  represent   the horizontal  and   vertical positions, respectively. The vehicle dynamics is  described by
$A = [\bar A_{11}, 0_{2\times 2}; 0_{2\times 2}, \bar A_{22}]$, where $\bar A_{11} = \bar A_{22} = [0, 0; 1, 0]$. This  represents the vehicle  undergoing  uniform liner motion with noise disturbances.
The sensor network consists of six sensors to measure the vehicle's state, and its  communication topology
is   illustrated   in Fig.~\ref{fig communication topology}.}

\begin{figure}[!htb]
\centering
{\includegraphics[width=1.6in]{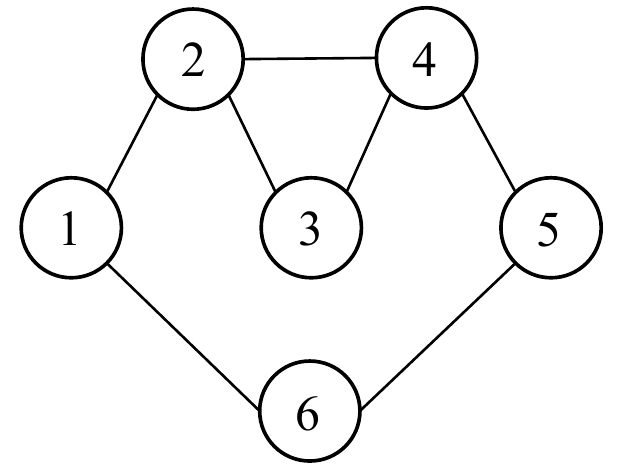}}
\caption{Illustration of  the communication topology.}
\label{fig communication topology}
\end{figure}

%
The process noise covariance is  given  as
~~$Q = \text{diag}\{0.03, 0.03, 0.03, 0.03\}$. Two kinds of sensors are deployed in the sensor network, and their measurement  matrices are defined as
$C_1 =  [0, 1, 0, 0]$, $C_2 =  [0, 2, 0, 0]$, $C_3 =  [0, 3, 0, 0]$, $C_4 =  [0, 0, 0, 1]$, $C_5 =  [0, 0, 0, 2]$,  and $C_6 =  [0, 0, 0, 3]$.
The measurement noise covariances  are set as $R_i = 0.2, ~ i=1,\ldots,6$.
The initial state is  $x(0) = [0.2; 1; 0.2; 1]$,  and the initial covariance is
$\Sigma_{i,e}(0) = \text{diag}\{0.1, 0.1, 0.1, 0.1\}$.
The Monte Carlo method  is adopted to evaluate the performance of the  distributed filter, and the mean square error (MSE)  is utilized  as a metric, given by $\text{MSE}(t) =  \frac{1}{MN}\sum^{N}_{i=1}\sum^{M}_{l=1}\Vert \hat x^{(l)}_{i,u}(t)-x^{(l)}(t) \Vert^2_2,$
where  $N$ is the sensor number  and $M$ is the trail number.

To demonstrate  the theoretical results, three  cases are considered. Case 1 validates the  theoretical theories of  the bounds of the estimation error covariance, the bounds of the nominal performance index, the estimation error covariance, and  the collective observability condition. Case 2 illustrates the divergence of the distributed filter caused by an  incorrect process noise covariance. Case 3 shows  the relations between $\Sigma_u(t)$ and $\Sigma_e(t)$.

Case 1:   Consider  the bounds of the estimation error covariance.
The vehicle dynamics is  described by
$A = [\bar A_{11}, 0; 0, \bar A_{22}]$, where $\bar A_{11} = \bar A_{22} = [-0.1, 0; 1, -0.1]$,  $A_u = A + 0.1 I_4$,
$C_{1,u} =  [0, 1.1, 0, 0]$, $C_{2,u} =  [0, 2.1, 0, 0]$, $C_{3,u} =  [0, 3.1, 0, 0]$, $C_{4,u} =  [0, 0, 0, 1.1]$, $C_{5,u} =  [0, 0, 0, 2.1]$,  $C_{6,u} =  [0, 0, 0, 3.1]$, $Q_u = 0.05 I_4$, and $R_{i,u} = 0.3, ~i=1,\ldots,6$.
%
%
Fig. \ref{fig_as_gamma}   displays    $\text{Tr}(\bar \Sigma_e)$,
the upper bound~1 of $\text{Tr}(\bar\Sigma_e)$  ($\text{Tr}(\bar\Sigma_u)+\rho(\gamma_u)$) referring to~(\ref{eq theorem bounds rho}),
the upper bound~2 of $\text{Tr}(\bar\Sigma_e)$  ($\text{Tr}(\bar\Sigma_u)+\rho(\gamma_u-0.5)$) referring to~(\ref{eq theorem bounds rho}) and (\ref{eq theorem gamma asymptotic}), the lower bound of $\text{Tr}(\bar\Sigma_u)$ referring  to (\ref{eq lower bound of barsigmau}), and MSE with the  increasing  consensus parameter $\gamma_u$. It is shown that
\begin{enumerate}
\item  The curves of   MSE   and   $\text{Tr}(\bar \Sigma_e)$  exhibit nearly overlapping,  demonstrating  the  correctness    of   the theoretical value $\bar \Sigma_e$.
\item  As  the consensus parameter $\gamma_u$ increases,  five indices decline. Moreover, both   MSE and    $\text{Tr}(\bar \Sigma_e)$  are constrained by the  upper   bound~1 and the upper bound~2,  and the asymptotic decay of the upper bound is demonstrated,
    as established  in Theorem~\ref{theorem bounds estimation error covariance}.
%

\end{enumerate}

\begin{figure}[!htb]
\centering
{\includegraphics[width=3.2in]{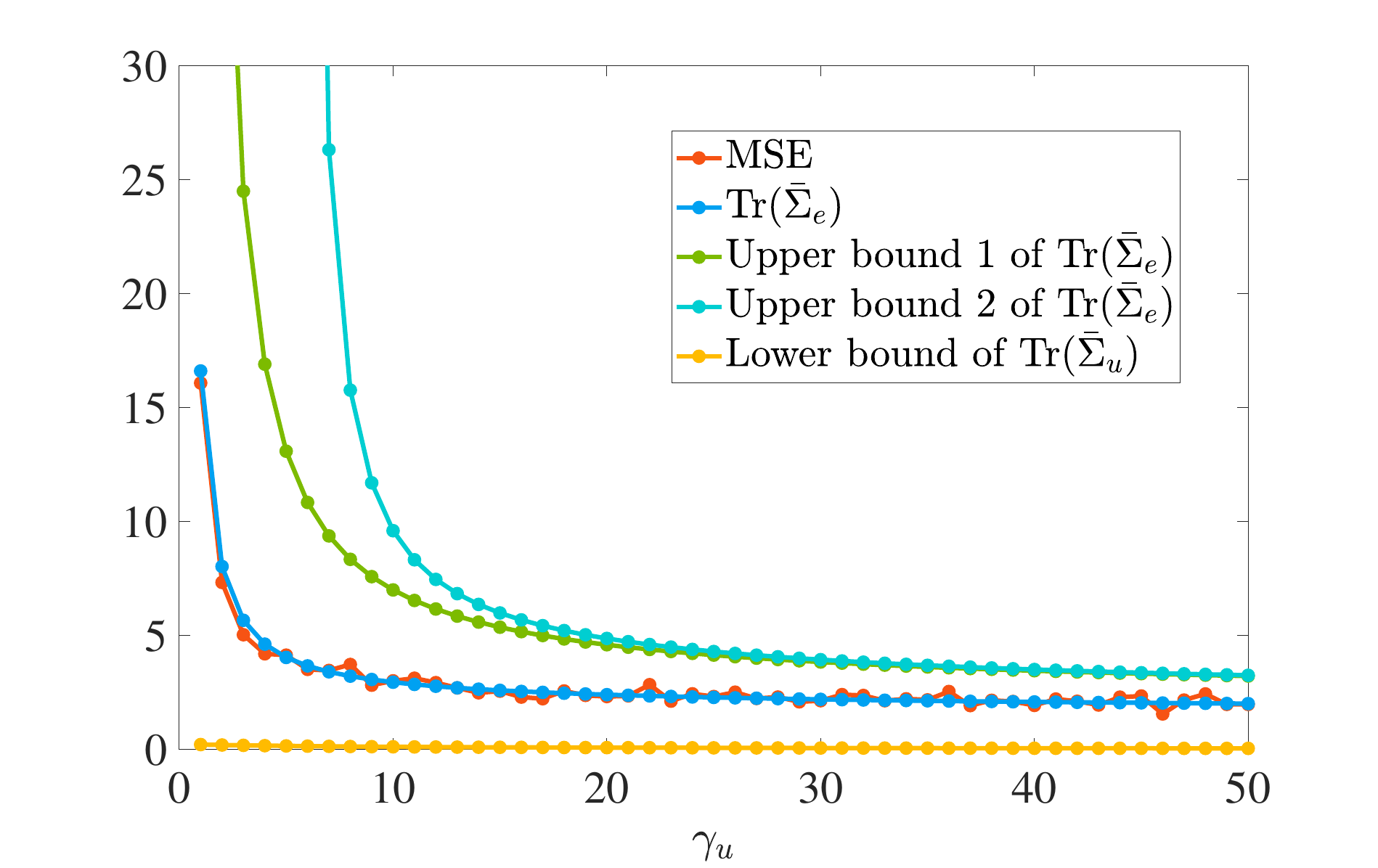}}
\caption{$\text{MSE}$,  $\text{Tr}(\bar \Sigma_e)$, the upper bound~1 of $\text{Tr}(\bar\Sigma_e)$,  the upper bound~2 of $\text{Tr}(\bar\Sigma_e)$, and the lower bound of $\text{Tr}(\bar\Sigma_u)$
with the increasing consensus parameter  in Case~1.}
\label{fig_as_gamma}
\end{figure}

\begin{figure}[!htb]
\centering
{\includegraphics[width=3.2in]{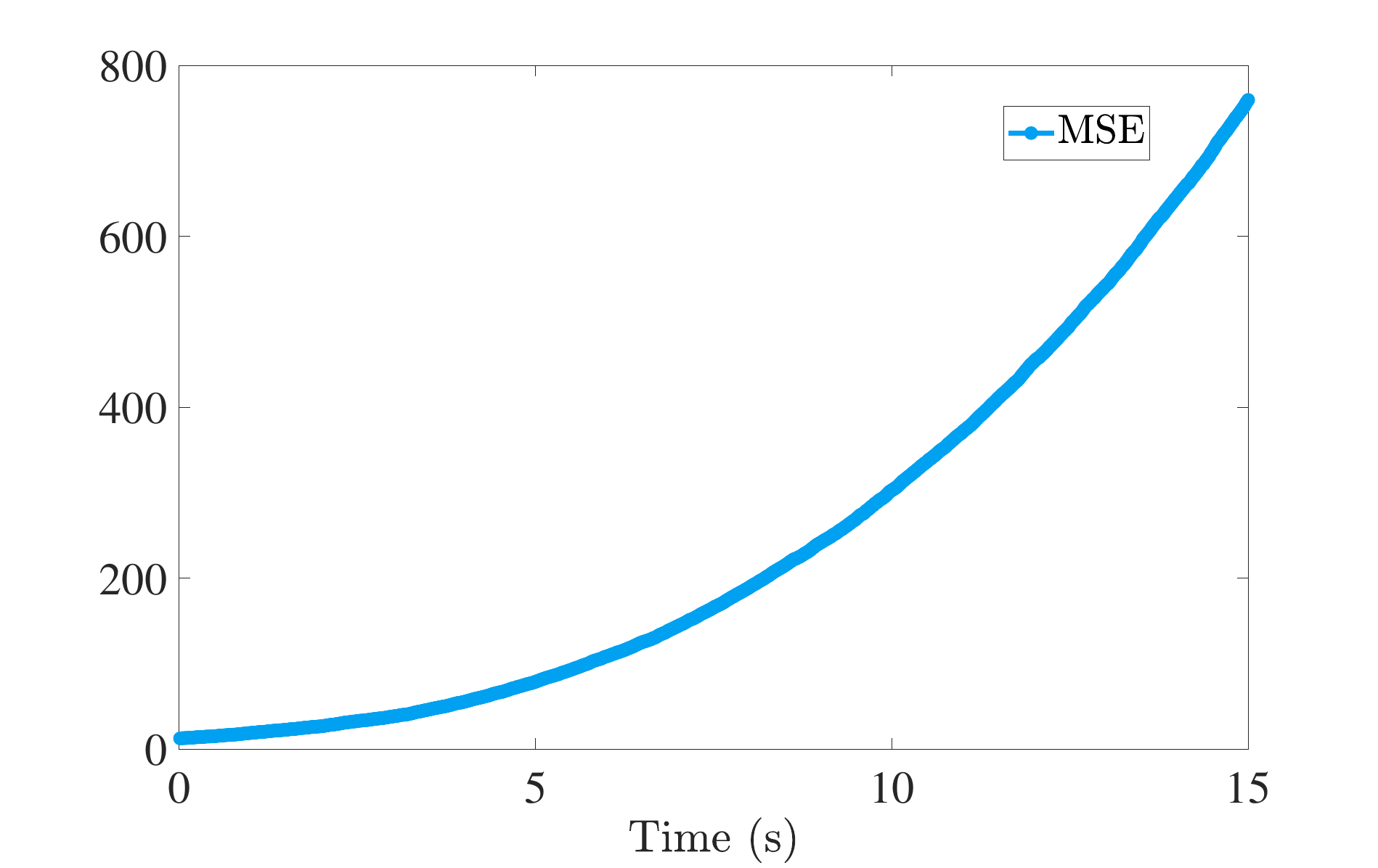}}
\caption{Illustration of the divergence caused by  the incorrect  covariance with the increasing time in Case 2.}
\label{fig_divergence}
\end{figure}

Case 2:   To verify  the effectiveness of
Theorem  \ref{theorem incorrect Q divergence},
the divergence of the  distributed   filter  resulting from   the  incorrect  covariance are considered.  Let us establish
the nominal process noise covariance  as  $Q_u = \text{diag}\{0, 0.03, 0.03, 0.03\}$, and set the consensus parameter  as   $\gamma_u = 10$.  First,  compute the eigenvalues and eigenvectors of  $A_u$.
All eigenvalues of  $A^T_u$   are  $0$, with  the corresponding eigenvectors are  $[1, 0, 0, 0]^T$, $[-1, 0, 0, 0]^T$, $[0, 0, 1, 0]^T$, and $[0, 0, -1, 0]^T$.
Consider the condition  $A^T_ue =rje$    and  $Q_ue=0$. From this,  it can be deduced that $r=0$, and $e=[1, 0, 0, 0]^T$ or  $e = [-1, 0, 0, 0]^T$.  Ultimately,  it can be calculated that  $\mathcal{A}^T_u(1_N \otimes e) = 0$.  Fig. \ref{fig_divergence}
shows  MSE of the distributed filter, and it can be observed   that  the actual estimation error has diverged due to utilizing  an incorrect  process noise covariance (See Theorem~\ref{theorem incorrect Q divergence}).

Case 3:  Consider the relations between $\Sigma_u(t)$ and $\Sigma_e(t)$.    The nominal parameters are set as $A_u = A$, $C_{i, u}= C_i$,   $Q_u = \text{diag}\{0.1, 0.1, 0.1, 0.1\}$, $R_{i,u} = 0.3, ~i= 1, 4$, and $R_{i,u} = 0.2, ~i= 2, 3, 5, 6$. The initial estimation error covariance and the initial performance index are chosen as   the identity matrix.  As stated in  (\ref{eq delta D}), it follows~$\Delta D>0$.

\begin{figure}[!htb]
\centering
{\includegraphics[width=3.2in]{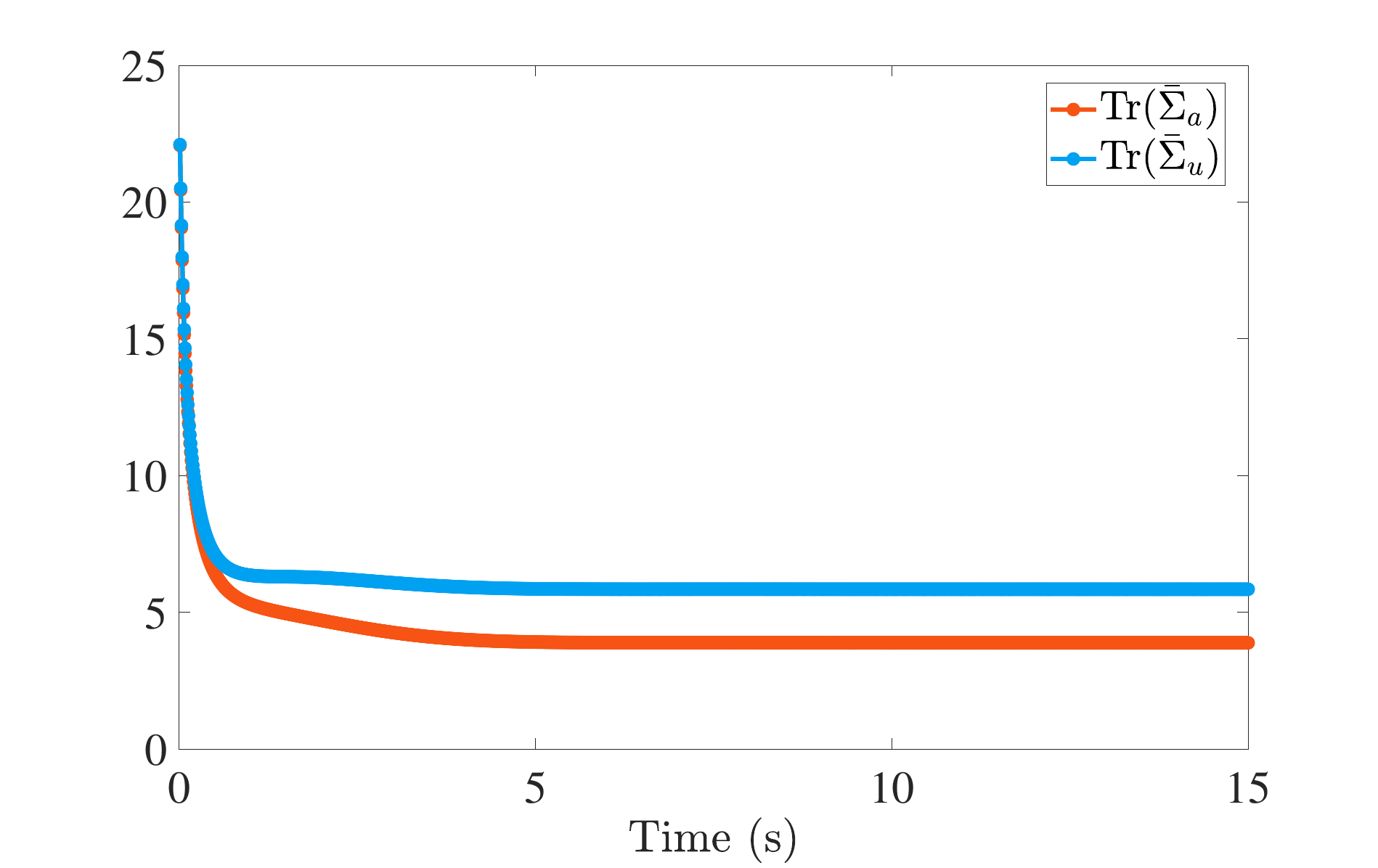}}
\caption{Illustration of the relations between  $\Sigma_u(t)$ and $\Sigma_e(t)$ with the increasing time in Case 3.}
\label{fig_deltaDge0}
\end{figure}

Fig. \ref{fig_deltaDge0}  exhibits
the trace of the estimation error covariance and  the nominal  performance index.  This figure illustrates  that the nominal  performance index  provides an upper  bound   of the estimation error covariance, which aligns  with the theoretical results (See Theorem \ref{theorem relation sigmau sigmaa}). It means that the estimation error covariance  of the distributed filter can be evaluated by utilizing  the nominal performance index, if the  nominal  parameters are appropriately designed.

\section{Conclusions}\label{sec conclusion}

In this paper, the performance analysis of the distributed continuous-time filter  in the presence of the modeling errors is conducted  from  two primary  aspects.  On the one hand,
  the convergence condition  and  the corresponding  convergence  analysis of   two performance indices   in the presence of  the modeling errors are provided. Then,   it is  demonstrated that   an incorrect noise covariance can lead to the divergence of the distributed filter.  On the other hand, this paper focuses on  the  performance evaluation  from the nominal performance index and the estimation error covariance.  The  bounds of the estimation error   are  derived by utilizing  the nominal performance index and the  model deviation information.   The relative magnitude relations between them are also presented. These results  shed light on
   the performance of the distributed filter in the presence of the modeling errors  and  provide guidance for  engineering applications.  In the future, we aim to   develop the corresponding distributed  algorithms  to handle the effects of the modeling errors.

\begin{appendices}

\section{Graph Theory}\label{sec graph theory}
A sensor network's nodes and  communication links can be represented as  the communication topology ${\mathcal{G}}(\mathcal{V},\mathcal{E})$, where the node set and the  edge set  are denoted as  $\mathcal{V}=\{1,2,\ldots,N\}$ and $\mathcal{E} \subseteq \mathcal{V} \times \mathcal{V}$, respectively.
For  $i, j \in \mathcal{V}$,  if the information can be transmitted from node $j$ to node $i$,
node $j$ is  called a neighbor of node $i$, denoted as
$(j,i)$.
The  neighbor set  of node $i$ is  represented   as  $\mathcal{N}_i =\{j|(j,i)\in \mathcal{V}\}$,  and  $\vert\mathcal{N}_i\vert$ is  the  cardinality of the neighbors of node $i$.
The adjacent matrix is   $S = [s_{ij}]_{N\times N}$, where
 $s_{ij} = 1$ if $(j,i)\in \mathcal{E}$, and   $s_{ij} = 0$ otherwise.
 The  Laplacian matrix  is  defined as   $\mathcal{L} = D-S$,  where $D = \text{diag}\{\vert\mathcal{N}_1\vert,\ldots,\vert\mathcal{N}_N\vert \}$. If $(i,j)\in \mathcal{E}$  implies    $(j,i) \in \mathcal{E}$, the edge $(i,j)$ is  called  undirected. If every edge  is undirected, the  communication graph is termed  undirected.   A directed path  from node $i_1$ to node $i_m$ exists in the graph $G$,
  if there is a sequence of connected edges $(i_k, i_{k+1}), k=1,\ldots,m-1$.
 The undirected communication graph is   connected if there exists a path between every two   nodes.

\section{Supporting  Lemmas}

\begin{lemma}\cite{soderlind2006logarithmic}\label{lemma mu logarithmic}
For matrices  $A$ and $B$ with the appropriate dimensions,   the following   properties of the logarithmic norm $\mu(\cdot)$  hold:
\begin{enumerate}
\item  $\mu(\gamma A) = \gamma \mu(A)$ for scalar $\gamma>0$,
\item  $\mu(A+B) \leq \mu(A) + \mu(B)$,
\item  $\Vert e^{tA}\Vert \leq e^{t\mu(A)}$ for $t\geq 0$.
\end{enumerate}

\end{lemma}

\begin{lemma}\cite{toda1978performance}\label{lemma tr AB trA B}
For matrices  $A$ and $B$ with the appropriate dimensions,   if
$A>0$ and $B\geq 0$, then  $\text{Tr}(AB)\leq  \text{Tr}(A) \bar\sigma(B) \leq  \text{Tr}(A)\text{Tr}(B),$ and $\text{Tr}(A^{-1}B)\geq (\bar\sigma(A))^{-1} \text{Tr}(B) \geq  [\text{Tr}(A)]^{-1}\text{Tr}(B).$
%
%
\end{lemma}

\begin{lemma}\cite{zhou1998essentials}\label{lemma sylvester equation}
For matrices  $A$, $B$, and $C$ with the appropriate dimensions,   consider the  Sylvester equation  $AX+XB =C.$
There exists a unique solution $X$ if and only if $\lambda_i(A)+\lambda_j(B)\neq 0$ for $i = 1,\ldots,m$ and $j=1,\ldots,n$.
\end{lemma}

\end{appendices}

\bibliographystyle{plain}        
\bibliography{autosam}           

\end{document}